\pdfoutput=1 
\documentclass[reqno, a4paper]{amsart}
\usepackage{amsmath}
\usepackage{amssymb}
\usepackage{amsthm}

\usepackage[english]{babel}
\usepackage[utf8]{inputenc}

\usepackage{natbib}

\usepackage[scale=0.9]{geometry}

\usepackage{subfig}
\usepackage{graphicx}

\usepackage{hyperref}

\usepackage{mathbbol}
\usepackage{bm} 
\usepackage{MnSymbol} 

\usepackage{units}
\usepackage{tensor}
\usepackage{accents}

\usepackage{enumitem}

\usepackage{lineno}

\usepackage{todonotes}

\usepackage{xcolor}
\usepackage{mdframed}
\usepackage{newfloat}

\usepackage{multicol}

\usepackage{comment}

\usepackage{booktabs}

\usepackage{placeins}

\DeclareMathOperator{\divergence}{div}

\DeclareMathOperator{\Tr}{Tr}

\newcommand{\bydefinition}{\mathrm{def}}
\newcommand{\traceless}[1]{{#1}_{\delta}}

\newcommand{\diff}{\mathrm{d}}

\renewcommand{\vec}[1]{\ensuremath{\mathbf{#1}}}

\makeatletter
\@ifpackageloaded{bm}%
{\renewcommand{\vec}[1]{\ensuremath{\bm{#1}}}%
}{%
\relax
}
\makeatother
\newcommand{\tensorq}[1]{\ensuremath{\mathbb{#1}}}      
\newcommand{\tensorc}[1]{\ensuremath{\mathrm{#1}}}      

\newcommand{\transpose}[1]{#1^\top}

\newcommand{\inverse}[1]{#1^{-1}}

\newcommand{\identity}{\ensuremath{\tensorq{I}}}

\newcommand{\cstress}{\tensorq{T}}
\newcommand{\cstressc}{\tensorc{T}}

\makeatletter
\@ifpackageloaded{bm}%
{%
}{%

}

\@ifpackageloaded{bm}%
{%
 
}{%

}

\@ifpackageloaded{bm}%
{%
}{%

}
\makeatother

\newcommand{\generictensor}{{\tensorq{A}}}

\newcommand{\dcstresssymb}{\traceless{\cstress}}

\newcommand{\vecv}{\ensuremath{\vec{v}}}
\newcommand{\gradv}{\ensuremath{\nabla \vecv}}

\newcommand{\gradsym}{\ensuremath{\tensorq{D}}}

\newcommand{\bvec}[1]{\vec{e}_{#1}} 

\newcommand{\bvecx}{\bvec{\hat{x}}}
\newcommand{\bvecy}{\bvec{\hat{y}}}
\newcommand{\bvecz}{\bvec{\hat{z}}}

\newcommand{\vhatp}[1][\vecvc]{{#1}^{\hat{\varphi}}}

\newcommand{\hatx}{\hat{x}}
\newcommand{\haty}{\hat{y}}

\newcommand{\cobvecn}[1]{\vec{g}_{\hat{#1}}} 

\newcommand{\R}{\ensuremath{{\mathbb R}}}
\newcommand{\pd}[2]{\ensuremath{\frac{\partial {#1}}{\partial {#2}}}}

\newcommand{\dd}[2]{\ensuremath{\frac{\diff {#1}}{\diff {#2}}}}

\makeatletter
\@ifpackageloaded{tensor}
{

}{%

}
\makeatother

\makeatletter
\@ifpackageloaded{tensor}
{

}{%

}
\makeatother

\newcommand{\norm}[2][]{\ensuremath{\left\|#2\right\|_{#1}}}
\newcommand{\absnorm}[1]{\ensuremath{\left|#1\right|}}

\makeatletter
\@ifundefined{volume}{%
}%
{%
}
\makeatother

\newcommand{\cvolumee}{\diff \mathrm{v}}

\makeatletter
\@ifpackageloaded{MnSymbol} 
{
 
}{%
 
}
\makeatother
\newcommand{\vectordot}[2]{\ensuremath{#1 \bullet #2}}

\newcommand{\tensorf}[1]{{\mathfrak{#1}}}

\def\X {{\boldsymbol X}}

\def\R {\mathbb{R}}

\def \beq{\begin{equation}}
\def \eeq{\end{equation}}
\def \ba{\begin{array}}
\def \ea{\end{array}}

\title[Numerical scheme for non-monotone constitutive relations]{Numerical scheme for simulation of transient flows of non-Newtonian fluids characterised by a non-monotone relation between the symmetric part of the velocity gradient and the Cauchy stress tensor}

\author{Adam Jane\v{c}ka}
\email{janecka@karlin.mff.cuni.cz}

\address{
  Faculty of Mathematics and Physics\\
  Charles University\\
  Sokolovsk\'a 83\\
  Praha 8 -- Karl\'{\i}n\\
  CZ 186\;75\\
  Czech Republic
}

\author{Josef M\'{a}lek}
\email{malek@karlin.mff.cuni.cz}

\address{
  Faculty of Mathematics and Physics\\
  Charles University\\
  Sokolovsk\'a 83\\
  Praha 8 -- Karl\'{\i}n\\
  CZ 186\;75\\
  Czech Republic
}

\author{V\'{\i}t Pr\r{u}\v{s}a}
\email{prusv@karlin.mff.cuni.cz}

\address{
  Faculty of Mathematics and Physics\\
  Charles University\\
  Sokolovsk\'a 83\\
  Praha 8 -- Karl\'{\i}n\\
  CZ 186\;75\\
  Czech Republic
}
\author{Giordano Tierra}
\email{gtierra@temple.edu}

\address{
  Department of Mathematics\\
  Temple University\\ 
  Philadelphia, PA 19122\\
  United States of America\\
}

\date{\today}
 
\keywords{non-Newtonian fluids, implicit constitutive relations, non-monotone constitutive relations, unsteady flow, finite element method}
\subjclass[2000]{%
  76D99, 74A20, 65M60}
\thanks{This research was supported by ERC-CZ project LL1202 funded by Ministry of Education, Youth and Sports of the Czech Republic. Josef M\'alek and V\'{\i}t Pr\r{u}\v{s}a acknowledge the support of the Czech Science Foundation project 18-12719S. Adam Jane\v{c}ka acknowledges the support of project 260449/2018 ``Student research in the field of physics didactics and mathematical and computer modelling''}

\numberwithin{equation}{section}
\let\cite\citet
\newtheorem{lemma}{Lemma}

\begin{document}

\begin{abstract}
  We propose a numerical scheme for simulation of transient flows of incompressible non-Newtonian fluids characterised by a non-monotone relation between the symmetric part of the velocity gradient (shear rate) and the Cauchy stress tensor (shear stress). The main difficulty in dealing with the governing equations for flows of such fluids is that the non-monotone constitutive relation allows several values of the stress to be associated with the same value of the symmetric part of the velocity gradient. This issue is handled via a reformulation of the governing equations. The equations are reformulated as a system for the triple pressure-velocity-apparent viscosity, where the apparent viscosity is given by a scalar implicit equation. We prove that the proposed numerical scheme has---on the discrete level---a solution, and using the proposed scheme we numerically solve several flow problems.


\end{abstract}

\maketitle

\addtocontents{toc}{\protect\begin{multicols}{2}} 
\tableofcontents


\section{Introduction}
\label{sec:introduction}
The response of non-Newtonian fluids is mathematically described in terms of a constitutive relation that links the Cauchy stress tensor~$\cstress$ and kinematical variables such as the symmetric part of the velocity gradient~$\gradsym$. In the case of incompressible non-Newtonian fluids, the stress is decomposed to the traceless part $\traceless{\cstress} = _{\bydefinition} \cstress - \frac{1}{3} \Tr \cstress$ and the spherical part $-p\identity$,
\begin{equation}
  \label{eq:1}
  \cstress = - p\identity + \traceless{\cstress},
\end{equation}
and a specific non-Newtonian fluid is usually characterised by an explicit constitutive equation of the type
\begin{equation}
  \label{eq:2}
  \traceless{\cstress} = \tensorf{f}(\gradsym),
\end{equation}
where $\tensorf{f}$ is a \emph{monotone} tensorial function. However, adequate description of the response of various non-Newtonian fluids requires one to consider \emph{non-monotone} tensorial functions $\tensorf{f}$ in~\eqref{eq:2}, see for example~\cite{david.j.filip.p:phenomenological} and \cite{galindo-rosales.fj.rubio-hernandez.fj.ea:apparent}. The non-monotone response also seems to be crucial in modelling complex non-Newtonian phenomena such as \emph{shear banding}, see \cite{fardin.ma.ober.tj.ea:potential} or \cite{divoux.t.fardin.ma.ea:shear}.

More importantly, several fluids have been reported to exhibit the behaviour that \emph{do not fit} into the framework~\eqref{eq:2}, but one can still formulate the constitutive relation as an \emph{algebraic relation} between $\traceless{\cstress}$ and $\gradsym$. In particular, constitutive relations for these fluids can take the form
\begin{equation}
  \label{eq:3}
  \gradsym = \tensorf{g}(\traceless{\cstress}),
\end{equation}
where $\tensorf{g}$ can be again a \emph{non-monotone} function. Moreover, one can also think about general implicit relations of the type
\begin{equation}
  \label{eq:22}
  \tensorf{h}(\gradsym, \traceless{\cstress}) = \tensorq{0},
\end{equation}
where $\tensorf{h}$ is a tensorial function. One of the very first observations of the fluid response that could be characterised by~\eqref{eq:3} is due to~\cite{boltenhagen.p.hu.y.ea:observation}, and the amount of experimental or theoretical works concerning the non-monotonous response of the type~\eqref{eq:3} has been growing since then, see for example~\cite{perlacova.t.prusa.v:tensorial}, \cite{janecka.a.prusa.v:perspectives}, \cite{rajagopal.kr.saccomandi.g:novel} or \cite{janecka.a.pavelka.m:non-convex} for further references.  

Typically, the non-monotone behaviour of $\tensorf{g}$ exhibits itself as an S-shaped curve in shear stress/shear rate plot, see Figure~\ref{fig:sshaped}. If one wants to avoid ``multivalued'' relations of the type~\eqref{eq:2}, then is clear that the non-monotonicity of $\tensorf{g}$ prevents one to invert~\eqref{eq:3}, and write constitutive relation~\eqref{eq:3} in the form~\eqref{eq:2}. Consequently, the class of fluids with constitutive relation of the type~\eqref{eq:3} substantially differs from the class of constitutive relations of the type~\eqref{eq:2}.

\begin{figure}[ht]
  \centering
  \subfloat[Standard way of thinking about constitutive relation, $\traceless{\cstress} = \tensorf{f}(\gradsym)$. Shear stress is sought as a function of shear rate.]{\includegraphics[width=0.35\textwidth]{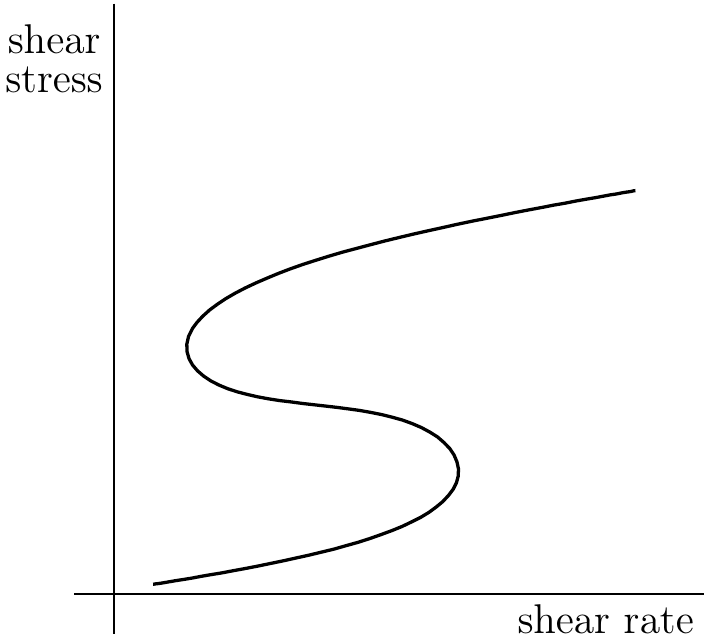}}
  \qquad
  \subfloat[Alternative way of thinking about constitutive relations, $\gradsym = \tensorf{g}(\traceless{\cstress})$. Shear rate is sought as a function of shear stress.]{\includegraphics[width=0.35\textwidth]{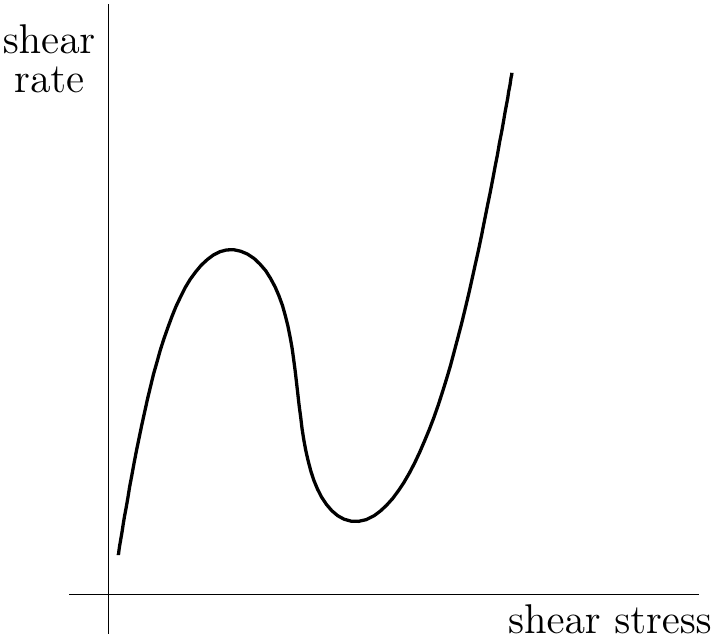}}
  \caption{S-shaped curve in shear rate/shear stress plot and shear stress/shear rate plot. In the standard way of thinking about constitutive relation, shear stress in not a function of shear rate. If the axes are rotated and one plots shear stress versus shear rate, then the shear rate is a function of the shear stress. The formulation of the constitutive relation as~\eqref{eq:3} instead of~\eqref{eq:2} is clearly more suitable.}
  \label{fig:sshaped}
\end{figure}

Flows of fluids with a non-monotone constitutive relation of the type~\eqref{eq:3} have been to our best knowledge investigated only in special geometries, where the corresponding system of governing equations reduces to a system of ordinary differential equations, see for example~\cite{malek.j.prusa.ea:generalizations}, \cite{roux.c.rajagopal.kr:shear}, \cite{narayan.spa.rajagopal.kr:unsteady}, \cite{srinivasan.s.karra.s:flow}, \cite{mohankumar.kv.kannan.k.ea:exact} and~\cite{fusi.l.farina.a:flow}. 
\emph{However, if one needs to investigate flows in more complex geometries, a suitable numerical scheme for solution of transient flow problems must be developed. Our aim is to address this issue.}

A particular constitutive relation that falls into the class~\eqref{eq:3} is the constitutive relation\footnote{The norm of a tensorial quantity $\generictensor$ is defined as the standard Frobenius norm, $\absnorm{\generictensor} =_{\bydefinition} \left( \Tr \left(\generictensor \transpose{\generictensor} \right))\right)^{\frac{1}{2}}$.}
\begin{equation}
  \label{eq:4}
  \gradsym = \left[\alpha \left(1 + \beta \absnorm{\traceless{\cstress}}^2 \right)^s + \gamma \right] \traceless{\cstress},
\end{equation}
that has been introduced by~\cite{roux.c.rajagopal.kr:shear}, see also~\cite{malek.j.prusa.ea:generalizations}. Symbols $\alpha$, $\beta$ denote positive constants, $\gamma$ is a nonnegative constant, and the exponent $s$ is a constant. If $s<-\frac{1}{2}$, then one can obtain, in general, a non-monotone response, see Figure~\ref{fig:consittutive-relation-scheme-a}, hence one is mainly interested in these values of the exponent $s$. (See \cite[Lemma 2.1]{roux.c.rajagopal.kr:shear} for a quantification of parameters range that lead to a non-monotone response.) Clearly, the development of a numerical scheme dealing with the simple constitutive relation~\eqref{eq:4} for $s<-\frac{1}{2}$ is a necessary step in the development of numerical schemes for more complex constitutive relations that belong into the class~\eqref{eq:3}.

\begin{figure}[ht]
  \centering
  \subfloat[\label{fig:consittutive-relation-scheme-a} Relation between the norms $\absnorm{\traceless{\cstress}}$ and $\absnorm{\gradsym}$.]{\includegraphics[width=0.45\textwidth]{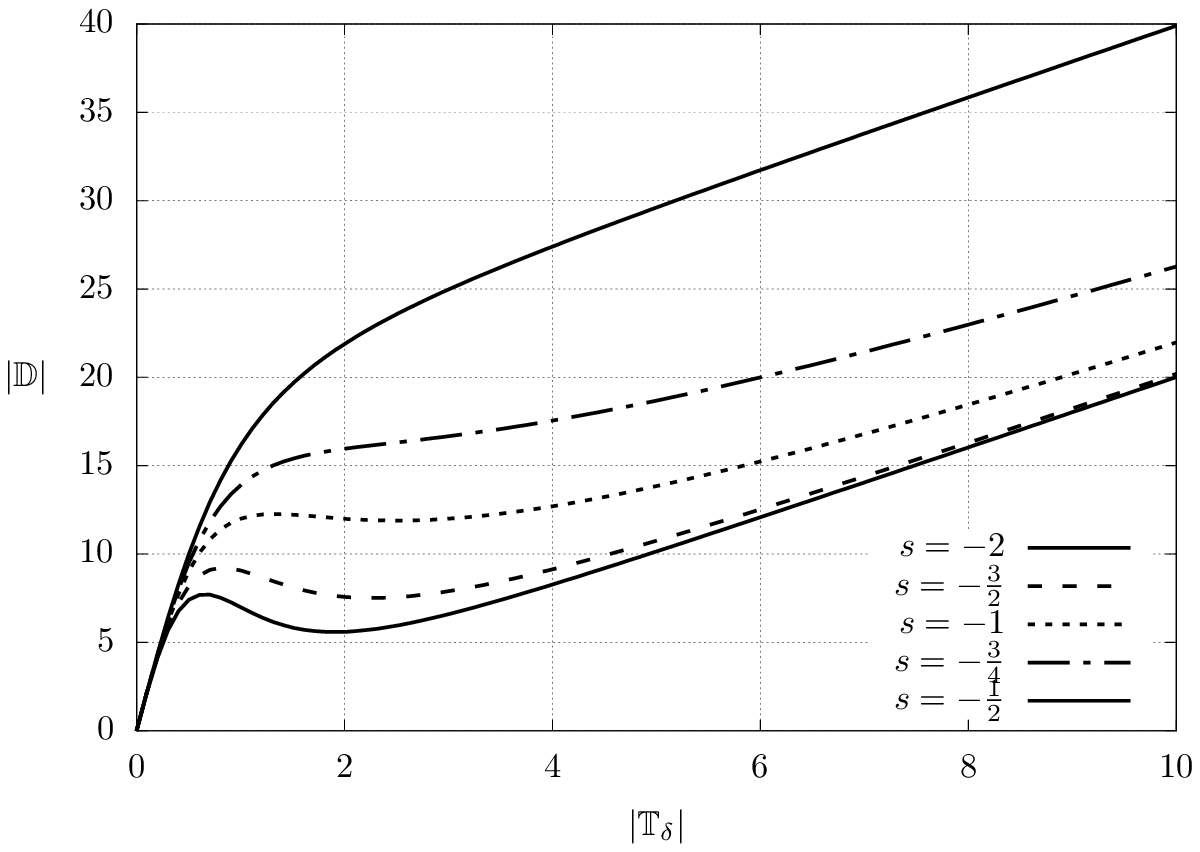}}
  \qquad
  \subfloat[\label{fig:consittutive-relation-scheme-b} Apparent viscosity $\tilde{\mu}$.]{\includegraphics[width=0.45\textwidth]{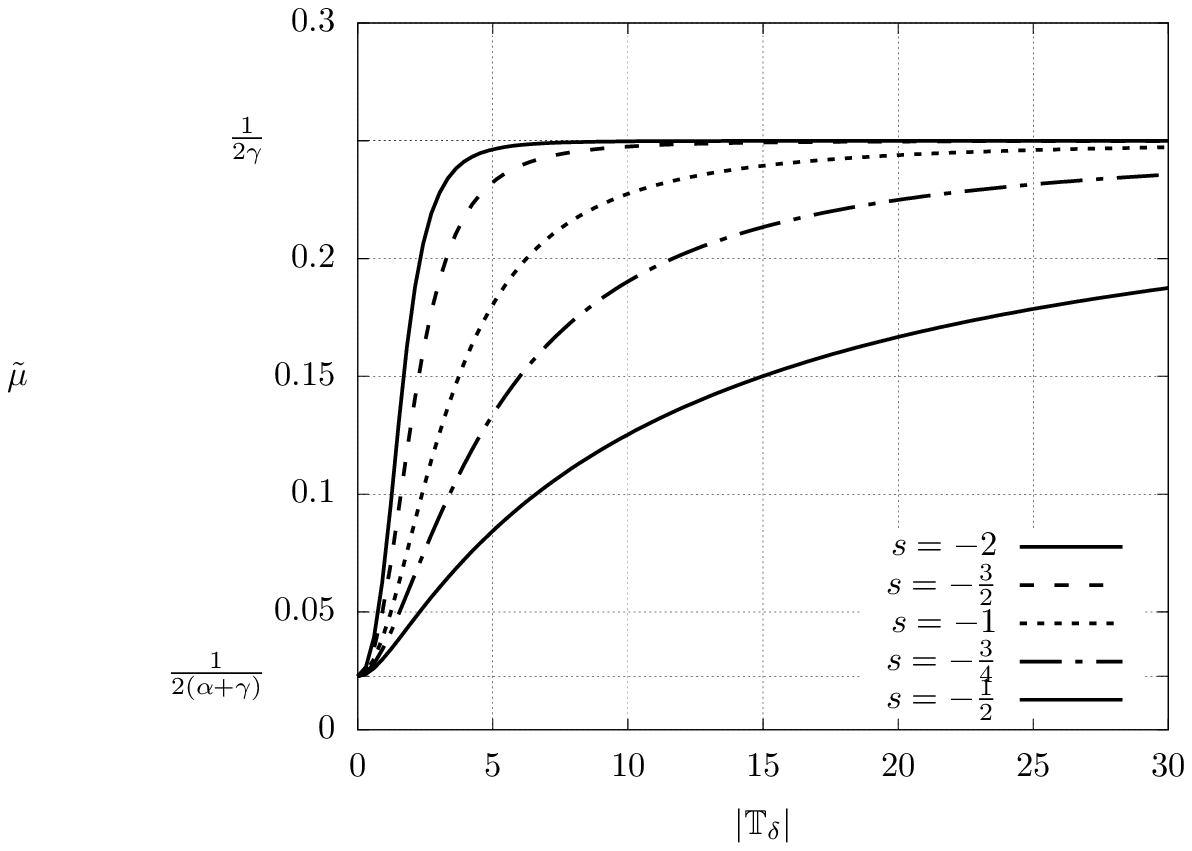}}
  \caption{Constitutive relation $\gradsym = [\alpha (1 + \beta \absnorm{\traceless{\cstress}}^2)^s + \gamma] \traceless{\cstress}$. Parameter values set to $\alpha=20$, $\beta=1$, $\gamma=2$. The exponent $s$ takes values $s \in \left\{-2, -\frac{3}{2}, -1, -\frac{3}{4}, -\frac{1}{2} \right\}$. The constitutive curve is non-monotone provided that $s < - \frac{1}{2}$ and $\frac{\gamma}{\alpha} < 2 \left(\frac{2s - 2}{2s + 1}\right)^{s-1}$, see~\cite{roux.c.rajagopal.kr:shear}. Consequently, for the given $\alpha$, $\beta$ and $\gamma$, the curve is non-monotone if, approximately, $s < -0.884341$.}
  \label{fig:consittutive-relation-scheme}
\end{figure}

The corresponding system of governing equations for an incompressible non-Newtonian fluid specified by constitutive relation~\eqref{eq:4} reads
\begin{subequations}
  \label{eq:5}
  \begin{align}
    \label{eq:7}
    \divergence \vec{v} &= 0, \\
    \label{eq:8}
    \rho \dd{\vec{v}}{t} &= - \nabla p + \divergence \traceless{\cstress} + \rho \vec{b}, \\
    \label{eq:9}
    \gradsym &= \left[\alpha \left(1 + \beta \absnorm{\traceless{\cstress}}^2 \right)^s + \gamma \right] \traceless{\cstress},
  \end{align}
\end{subequations}
where $\rho$ denotes the constant density, $\vec{v}$ the Eulerian velocity field, $\vec{b}$ the external body forces, and $\dd{}{t}$ stands for the material time derivative. The first difficulty in solving~\eqref{eq:5} for $s<-\frac{1}{2}$ is the fact that one can not, in general, invert~\eqref{eq:9} and express $\traceless{\cstress}$ as a function of $\gradsym$. (The constitutive curve is non-monotone provided that $s < -\frac{1}{2}$, and that the other parameter values satisfy inequality $\frac{\gamma}{\alpha} < 2 \left(\frac{2s - 2}{2s + 1}\right)^{s-1}$, see~\cite{roux.c.rajagopal.kr:shear}.) Consequently, system~\eqref{eq:5} can not be, in general, formulated as a system for the pressure-velocity pair  $(p, \vec{v})$. However, this is just a matter of a reformulation of the problem as a problem for the pressure-velocity-stress triple $(p, \vec{v}, \traceless{\cstress})$.

The key conceptual difficulty is the treatment of the constitutive relation~\eqref{eq:9}. The reason is that~\eqref{eq:9} admits for $s< -\frac{1}{2}$ multiple values of $\traceless{\cstress}$ to be associated with the same value of the symmetric part of the velocity gradient $\gradsym$. \emph{In what follows we focus exclusively on this most interesting case, that is we consider only $s < -\frac{1}{2}$}.

Unfortunately, the multiplicity issue prevents one from using most of the already available analytical and numerical results regarding initial and boundary value problems for systems of the type
\begin{subequations}
  \label{eq:10}
  \begin{align}
    \label{eq:11}
    \divergence \vec{v} &= 0, \\
    \label{eq:12}
    \rho \dd{\vec{v}}{t} &= - \nabla p + \divergence \traceless{\cstress} + \rho \vec{b}, \\
    \label{eq:13}
    \tensorf{f}(\traceless{\cstress}, \gradsym) &= \tensorq{0},
  \end{align}
\end{subequations}
where $\tensorf{f}(\traceless{\cstress}, \gradsym)$ is an implicit function. The available results, see \cite{bulcek.m.gwiazda.p.ea:on*2,bulcek.m.gwiazda.p.ea:on*3,bulcek.m.gwiazda.p.ea:on*4} and also~\cite{maringova.e.zabensky.j:on} for the proof of long-time and large-data existence of weak solution to~\eqref{eq:10} and similar systems, and \cite{stebel.j:finite}, \cite{diening.l.kreuzer.c.ea:finite} and~\cite{suli.e.tscherpel.t:fully} for the results concerning the discretised counterparts of~\eqref{eq:10}, are based on the fact that the equation $\tensorf{f}(\traceless{\cstress}, \gradsym) = \tensorq{0}$ defines a \emph{maximal monotone graph}. Although the maximal \emph{monotone} graph defined by~$\tensorf{f}$ can be possibly multivalued, such as in the case of Bingham fluid, see~\cite{bulcek.m.gwiazda.p.ea:on*3} and~\cite{hron.j.malek.j.ea:novel}, systems of the type~\eqref{eq:5} with \emph{non-monotone} response, that is~\eqref{eq:9} with $s < - \frac{1}{2}$, are not covered by the otherwise very general theory by~\cite{bulcek.m.gwiazda.p.ea:on*2,bulcek.m.gwiazda.p.ea:on*3}.

The numerical scheme for solution of~\eqref{eq:5} introduced below \emph{represents the first attempt to study systems of type~\eqref{eq:5}}. The proposed scheme does not fully answer the question on the existence of a solution to~\eqref{eq:5}, yet an important step is made. Namely, a discrete finite-dimensional nonlinear system that arises in the time-stepping of system~\eqref{eq:5} is shown to be solvable.

The work is organized as follows. In Section~\ref{sec:model} we reformulate system~\eqref{eq:5} as a nonlinear system for the pressure-velocity-apparent viscosity triple $(p, \vec{v}, \mu)$. The reformulation is the key step in the derivation of the numerical scheme. The tensorial constitutive relation~\eqref{eq:4} is effectively replaced by an implicit relation for a scalar quantity---the apparent viscosity---and the arising system shares some similarities with the standard Navier--Stokes system. Moreover, the apparent viscosity function introduced by this reformulation is monotone and bounded which allows one to obtain \emph{a priori} estimates.

Then, in Section~\ref{sec:scheme} we describe a numerical scheme for the solution of the governing equations, and we show that there exists a solution to the discretised counterpart of the governing equations.

In order to study the dynamical behaviour implied by the non-monotone constitutive relations, we introduce, see Section~\ref{sec:uq}, a reduced version of the problem. In the reduced problem we neglect the contributions from other effects like pressure (incompressibility) or convection, and we focus solely on the constitutive relation. (Note that the reduced problem can be seen as a heat conduction problem with non-monotonous heat flux versus temperature gradient constitutive relation, see~\cite{janecka.a.pavelka.m:non-convex} and references therein.) We introduce a variant of the proposed numerical scheme for the reduced problem, and we present several numerical experiments that document the behavior of the reduced system.

In Section~\ref{sec:simulations} we move forward and we solve the full problem~\eqref{eq:5}. Using the proposed numerical scheme we perform numerical experiments in two settings. First, we investigate the flow in the cylindrical Couette geometry, and, second, we investigate the flow in a channel with a narrowing. Finally, the conclusions of our work are stated in Section~\ref{sec:conclusion}.


\section{Reformulation of the problem in terms of apparent viscosity}
\label{sec:model}
Let us consider a bounded domain $\Omega\subset \R^d$ with $(d = 2, 3)$. For the sake of simplicity of the discussion let us further assume that no external body force is present, $\vec{b}=\vec{0}$. Then system~\eqref{eq:5} reads
\begin{subequations}
  \label{eq:16}
  \begin{align}
    \label{eq:17}
    \divergence \vecv &= 0, \\
    \label{eq:18}
    \rho \dd{\vecv}{t} &= - \nabla p + \divergence \dcstresssymb, \\
    \label{eq:19}
    \gradsym &= \left[\alpha \left(1 + \beta \absnorm{\dcstresssymb}^2 \right)^s + \gamma \right] \dcstresssymb,
  \end{align}
\end{subequations}
where the standard notation $\dd{\vecv}{t} = _{\bydefinition} \pd{\vecv}{t} + \left(\vectordot{\vecv}{\nabla}\right) \vecv$, $\gradsym =_{\bydefinition} \frac{1}{2} \left( \gradv + \transpose{\left( \gradv \right)} \right)$, $\absnorm{\dcstresssymb} =_{\bydefinition} \Tr \left( \dcstresssymb \transpose{ \dcstresssymb} \right)^{\frac{1}{2}}$ has been used. Note that since the trace of $\dcstresssymb$ is by definition equal to zero, the constitutive relation~\eqref{eq:19} in fact already enforces the incompressibility constraint~\eqref{eq:17}. We shall however keep~\eqref{eq:17} in the system, since the constitutive relation~\eqref{eq:19} will be soon reformulated. The price paid for the reformulation is that the identity~\eqref{eq:17} that otherwise automatically follows from~\eqref{eq:19} must be kept in the reformulated system in an explicit way. 

The system is supplemented with the initial and boundary conditions in the form
\begin{subequations}
  \label{eq:6}
  \begin{align}
    \label{eq:14}
    \left. \vecv(\vec{x}, t) \right|_{t=0}         &=\vecv_0(\vec{x}), \\
    \label{eq:15}
    \left. \vecv(\vec{x}, t) \right|_{\partial\Omega} &= \vec{0}. 
  \end{align}
\end{subequations}
(The zero Dirichlet boundary condition~\eqref{eq:15} is used in mathematical analysis of the the governing equations. In the numerical simulations we use a non-zero Dirichlet boundary condition.) The aim is to find the triple $(p, \vecv, \dcstresssymb)$ such that it solves~\eqref{eq:16} subject to~\eqref{eq:6}. 

Introducing the apparent viscosity $\mu$ by the formula
\begin{equation}
  \label{eq:20}
  \mu (\vec{x},t)=_{\bydefinition} \frac{1}{2} \frac{\absnorm{\dcstresssymb (\vec{x},t)}}{\absnorm{\gradsym (\vec{x},t)}},
\end{equation}
we see that the apparent viscosity $\mu(\vec{x},t)$ can be rewritten as a function of the traceless part of the Cauchy stress tensor, that is $\mu(\vec{x},t) = \tilde{\mu}(\absnorm{\dcstresssymb(\vec{x},t)})$, where
\begin{equation}
  \label{def:viscosity}
  \tilde{\mu}(u)=_{\bydefinition} \frac{1}{2} \inverse{\left[\alpha\left(1 + \beta u^2\right)^s + \gamma\right]}.
\end{equation}
Note that if $s<-1/2$, then the apparent viscosity $\tilde{\mu}$ introduced in \eqref{def:viscosity} is a positive \emph{increasing and bounded function} of $u$ satisfying for all $u \in [0,+\infty)$ inequalities
\begin{equation}
  \label{eq:21}
  \frac1{2(\gamma+\alpha)} \leq \tilde{\mu}(u) \leq \frac1{2\gamma},
\end{equation}
see Figure~\ref{fig:consittutive-relation-scheme-b}. Further, constitutive relation~\eqref{eq:19} can be rewritten as $\traceless{\cstress} = \tilde{\mu}( \absnorm{\traceless{\cstress}}) \gradsym$, which implies that system~\eqref{eq:16} can be reformulated as a system
\begin{subequations}
  \label{model2}
  \begin{align}
    \label{eq:24}
    \divergence \vecv &= 0, \\
    \label{eq:25}
    \rho \dd{\vecv}{t} &= - \nabla p + \divergence \left(2 \mu \gradsym \right), \\
    \label{eq:26}
    \mu &= \frac{1}{2} \inverse{\left[\alpha\left(1 + 4 \beta \mu^2 \absnorm{\gradsym}^2\right)^s + \gamma\right]}
  \end{align}
\end{subequations}
for the triple $(p, \vecv, \mu)$. Note that the last equation is an implicit equation for the apparent viscosity $\mu$ in terms of $\gradsym$. This reformulation is a useful one, since the implicit equation is now an equation for a \emph{scalar} variable, while in~\eqref{eq:16} the implicit equation is an equation for a \emph{tensorial} variable.

At this stage, we can observe that taking the scalar product of~\eqref{eq:25} with~$\vecv$ in \eqref{eq:25} and using the incompressibility constraint~\eqref{eq:24}, we formally arrive to the identity
\begin{equation}
  \label{eq:39}
  \pd{}{t}\left( \frac{1}{2} \rho \absnorm{\vecv}^2 \right)
  +
  \divergence
  \left[
    \left(
      p
      +
      \frac{1}{2} \rho \absnorm{\vecv}^2
    \right)
    \vec{v}
    -
    2 \mu \gradsym
    \vec{v}
  \right]
  =
  -
  2 \mu \absnorm{\gradsym}^2.
\end{equation}
If the boundary condition reads $\left. \vec{v} \right|_{\Omega} = \vec{0}$, that is if the system is mechanically isolated, then the integration of~\eqref{eq:39} over $\Omega$ and the application of the Stokes theorem leads us to the formal balance of mechanical energy
\begin{equation}
  \label{contenergylaw}
  \dd{}{t} \left( \frac{1}{2} \int_\Omega \rho \absnorm{\vecv}^2 \, \cvolumee \right) + 2 \int_\Omega \mu \absnorm{\gradsym}^2 \, \cvolumee = 0.
\end{equation}
In particular, \eqref{contenergylaw}) implies that the kinetic energy is indeed dissipated. Further, due to \eqref{eq:21}, we see that~\eqref{contenergylaw} also implies \emph{a priori} estimates analogous to that available for the standard Navier-Stokes system,
\begin{equation}
  \label{eq:40}
  \vecv \in L^{\infty} (0,T;L^2(\Omega)^d) \cap L^2(0,T;W_0^{1,2}(\Omega)^d).
\end{equation}
Concerning the pressure, one expects at least for spatially periodic problem or for any kind of slip conditions, see~\cite{bulmal}, to recover the same estimates obtained in the theory of the Navier-Stokes equations, namely
\begin{equation}
  \label{eq:41}
   p \in L^2(0,T;L^2(\Omega)) \quad\text{if } d=2, \qquad \qquad
   p \in L^{5/3}(0,T;L^{5/3}(\Omega)) \quad\text{if } d=3,
 \end{equation}
where the pressure $p$ is suitably normalised, for example via the condition $\int_{\Omega} p \, \cvolumee = 0$.


\section{Numerical scheme}
\label{sec:scheme}
The aim is to propose a numerical scheme that can be used in solving~\eqref{model2} by finite elements in space and finite differences in time. We exploit the fact that unlike in~\eqref{eq:19} the implicit constitutive relation in~\eqref{model2} is replaced by a scalar implicit constitutive relation for the apparent viscosity.

Concerning the time discretisation, we assume a uniform partition of the time interval, $t_n = n\Delta t$, where $\Delta t > 0$ represents a fixed time step.  Moreover, since the main issue is the non-monotonicity of the constitutive relation, we for the sake of simplicity neglect the convective terms in the presentation of the numerical algorithm. However, the convective term is present in the numerical experiments reported later in Section~\ref{sec:simulations}.

Let ${\boldsymbol V}_h$, $P_h$ and $T_h$ be finite-dimensional spaces with bases $\left\{ \vec{\phi}_k \right\}_{k=1}^{N_{\vec{V}}}$, $\left\{ q_k \right\}_{k=1}^{N_{P}}$ and $\left\{ t_k \right\}_{k=1}^{N_{T}}$ respectively, while the pair ${\boldsymbol V}_h$, $P_h$ satisfies the standard Babu\v{s}ka--Brezzi condition, see for example~\cite{brezzi.f.fortin.m:mixed}.  In practice, ${\boldsymbol V}_h$, $P_h$ and $T_h$ denote the finite element spaces, related to a regular triangulation~$\mathcal{T}_h$ of the domain $\Omega$, see Section~\ref{sec:simulations} for the specification of the finite lement spaces used in the numerical experiments. Let us assume that the solution at time $t_n$ denoted as $(\vecv^n, p^n, \mu^n) \in \vec{V}_h\times P_h\times T_h$ is known, and let us compute the solution at time $t_{n+1}$ denoted as $(\vecv^{n+1},p^{n+1},\mu^{n+1}) \in \vec{V}_h\times P_h\times T_h$ as a solution of the following system of nonlinear algebraic equations
\begin{subequations}
  \label{nonlinear_scheme}
  \begin{align}
    \label{eq:nonlinear_scheme1}
    \rho\left(\frac{\vecv^{n+1} - \vecv^n}{\Delta t}, \bar{\vecv} \right)
    + \left( 2 \mu^{n+1} \gradsym^{n+1}, \overline{\gradsym} \right) 
    - \left( p^{n+1}, \divergence \bar{\vecv} \right)
    &= 0, \\
    \label{eq:nonlinear_scheme2}
    \left( \divergence \vecv^{n+1}, \bar{p} \right) &= 0, \\
    \label{eq:nonlinear_scheme3}
    \left( \mu^{n+1}, \bar{\mu} \right) 
    - \left( \frac{1}{2} \left[ \alpha \left( 1 + 4 \beta \absnorm{\mu^{n+1}}^2 \absnorm{\gradsym^{n+1}}^2 \right)^s + \gamma \right]^{-1}, \bar{\mu}\right) &= 0,  
  \end{align}
\end{subequations}
that must be satisfied for all base functions $(\bar{\vecv}, \bar{p}, \bar{\mu})$ in $\vec{V}_h \times P_h \times T_h$. Here the symbol $\overline{\gradsym}$ denotes the symmetric part of the gradient of $\bar{\vec{v}}$, $\overline{\gradsym} =_{\bydefinition} \frac{1}{2} \left( \nabla \bar{\vecv} + \transpose{\nabla \bar{\vecv}} \right)$, and the symbol $(a,b) = _{\bydefinition} \int_{\Omega} ab \, \cvolumee$ denotes the standard scalar product in the Lebesgue space $L^2(\Omega)$.

Since~\eqref{nonlinear_scheme} holds for all base functions $(\bar{\vecv}, \bar{p}, \bar{\mu})$ in $\vec{V}_h \times P_h \times T_h$, we also know that~\eqref{nonlinear_scheme} also holds if we set $(\bar{\vecv}, \bar{p}) = (\vecv^{n+1}, p^{n+1})$. This helps us to recover, in the case of boundary condition $\left. \vec{v} \right|_{\Omega} = \vec{0}$, the discrete version of the balance of energy~\eqref{contenergylaw},
\begin{equation}
  \label{discreteEL}
  \frac{1}{2} \rho \norm[L^2(\Omega)]{\vecv^{n+1}}^2
  +
  2 \Delta t \int_\Omega \mu^{n+1} \absnorm{\gradsym^{n+1}}^2
  \leq
  \frac{1}{2} \rho \norm[L^2(\Omega)]{\vecv^n}^2
  \leq
  \frac{1}{2} \rho \norm[L^2(\Omega)]{\vecv_{0,h}}^2,
\end{equation}
where the last inequality follows from the iteration of the first inequality with respect to $n$, and where~$\vecv_{0,h}$ denotes the approximation of the initial condition~\eqref{eq:14} in the space $\vec{V}_h$. Similarly, from~\eqref{eq:nonlinear_scheme3} and~\eqref{eq:21} we get
\begin{equation}
  \label{eq:42}
  \norm[L^2(\Omega)]{\mu^{n+1}} \leq \frac{\absnorm{\Omega}^{\frac{1}{2}}}{2 \gamma},
\end{equation}
where $\absnorm{\Omega}$ denotes the area/volume of the domain $\Omega$.

The question is whether the system of nonlinear algebraic equations has a solution. Using a standard lemma, we show that  that there exists at least one solution of~\eqref{nonlinear_scheme}. The standard lemma, which is proved for example in~\cite[Lemma~1.4, page~164, Chapter~II]{temam.r:navier-stokes*2}, states that

\begin{lemma}\label{temamLemma}
Let $\X$ be a finite dimensional Hilbert space with scalar product $(\cdot,\cdot)$ and norm $|\cdot|$, and let~$\mathcal{P}$ be a continuous mapping from $\X$ into itself. Assume that there exists $\theta>0$ such that
\begin{equation}
  \label{eq:45}
  (\mathcal{P}(\xi),\xi)>0 \quad \mbox{ for } |\xi|=\theta > 0.  
\end{equation}
Then, there exists $\xi\in \X$, $|\xi|\le \theta $, such that
\begin{equation}
  \label{eq:46}
  \mathcal{P}(\xi)=0.  
\end{equation}
\end{lemma}
Now we are ready to present the proof of the existence of a solution.
\begin{lemma}
There exists at least one solution to the nonlinear system of algebraic equations~\eqref{nonlinear_scheme}.
\end{lemma}
\begin{proof}
  The existence of a solution of \eqref{nonlinear_scheme} follows from Lemma~\ref{temamLemma}. First we need to define the appropriate function space, we set
  \begin{equation}
    \label{eq:43}
    \X_h
    =
    _{\bydefinition}
    \vec{V}_h\times P_h\times T_h,
  \end{equation}
  and we define the corresponding scalar product via
  \begin{equation}
    \label{eq:44}
    \Big((\tilde{\boldsymbol v}, \tilde p, \tilde\mu), (\bar{\boldsymbol v}, \bar p,\bar \mu) \Big)
    =
    _{\bydefinition}
    (\tilde{\boldsymbol v}, \bar{\boldsymbol v})
    +
    (\tilde p, \bar p)
    +
    (\tilde\mu,\bar \mu)
    ,
  \end{equation}
  where the round brackets denote the scalar product in $\vec{V}_h$, $P_h$ and $T_h$ respectively. The mapping $\mathcal{P}(\tilde{\boldsymbol v}, \tilde p, \tilde\mu)$ is defined in the standard way as
  \begin{equation}
    \label{eq:47}
    \mathcal{P}:
    \begin{bmatrix}
      \tilde{\vecv} \\
      \tilde{p} \\
      \tilde{\mu}
    \end{bmatrix}
    \mapsto
    \begin{bmatrix}
      \rho \frac{\tilde{\vec{v}} - \vecv^n}{\Delta t} + 2 \tilde{\mu} \tilde{\gradsym} + \nabla \tilde{p} \\
      \divergence \tilde{\vec{v}} \\
      \tilde{\mu} - \frac{1}{2} \left[ \alpha \left( 1 + 4 \beta \absnorm{\tilde{\mu}}^2 \absnorm{\tilde{\gradsym}}^2 \right)^s + \gamma \right]^{-1}
    \end{bmatrix}
    ,
  \end{equation}
  where the symbol $\tilde{\gradsym}$ denotes the symmetric part of the gradient of $\tilde{\vec{v}}$, $\tilde{\gradsym} =_{\bydefinition} \frac{1}{2} \left( \nabla \tilde{\vecv} + \transpose{\nabla \tilde{\vecv}} \right)$. The product $(\mathcal{P}(\tilde{\xi}),\bar{\xi})$ then reads
  \begin{multline}
    \label{eq:48}
    \Big(\mathcal{P}(\tilde{\boldsymbol v}, \tilde p, \tilde\mu),(\bar{\boldsymbol v}, \bar p,\bar \mu)\Big)
    =
    \rho\frac{1}{\Delta t}(\tilde{\boldsymbol v}, \bar {\boldsymbol v})
    -
    \rho\frac{1}{\Delta t}({\boldsymbol v}^n, \bar {\boldsymbol v})
    +
    2(\tilde{\mu}  \tilde{\gradsym}, \overline{\gradsym}) 
    +
    (\tilde{\mu}, \bar\mu)
    \\
    -
    \Big( \frac{1}{2} \big[\alpha(1 + 4 \beta \absnorm{\tilde{\mu}}^2 \absnorm{\tilde{\gradsym}}^2)^s + \gamma\big]^{-1},\bar\mu\Big)
    -
    (\tilde p, \divergence \bar {\boldsymbol v})
    +
    (\divergence \tilde{\boldsymbol v}, \bar p)
    ,
  \end{multline}
  hence for the product $$(\mathcal{P}(\tilde{\xi}),\tilde{\xi})$$ one has
  \begin{multline}
    \label{eq:49}
    \Big(\mathcal{P}(\tilde{\boldsymbol v}, \tilde p, \tilde\mu),(\tilde{\boldsymbol v}, \tilde p, \tilde\mu)\Big)
    \geq
    \rho \frac{1}{2 \Delta t} \norm[L^2(\Omega)]{\tilde{\boldsymbol v}}^2 
    -
    \rho \frac{1}{2 \Delta t} \norm[L^2(\Omega)]{\vecv^n}^2 
    +
    2
    \int_\Omega \tilde{\mu} \absnorm{\tilde{\gradsym}}^2 \, \cvolumee
    \\
    +
    \frac{1}{2}\norm[L^2(\Omega)]{\tilde{\mu}}^2 
    -
    \frac{1}{2}
    \norm[L^2(\Omega)]{ \left[ \frac{1}{2} \alpha(1 + 4 \beta \absnorm{\tilde{\mu}}^2 \absnorm{\tilde{\gradsym}}^2)^s + \gamma \right]^{-1}}^2,
  \end{multline}
  where we have used the Cauchy--Schwarz inequality. Clearly, if $\norm[L^2(\Omega)]{\tilde{\boldsymbol v}}^2$ and $\norm[L^2(\Omega)]{\tilde{\mu}}^2$ are large enough, we have
  \begin{equation}
    \label{eq:50}
    \Big(\mathcal{P}(\tilde{\boldsymbol v}, \tilde p, \tilde\mu),(\tilde{\boldsymbol v}, \tilde p, \tilde\mu)\Big) > 0.
  \end{equation}
  (Recall that~\eqref{eq:21} guarantees boundedness of the last term in~\eqref{eq:49}.) We can use Lemma~\ref{temamLemma} and conclude that the exists a triple $(\hat{\boldsymbol v}, \hat p, \hat\mu)\in\X_h$ such that $\mathcal{P}(\hat{\boldsymbol v}, \hat p, \hat\mu)=0$.   Consequently, due to the definition of~$\mathcal{P}$ the triple $(\hat{\boldsymbol v}, \hat p, \hat\mu)$ is a solution of~\eqref{nonlinear_scheme}. 
\end{proof}

\subsection{Iterative algorithm}
\label{sec:iterative-algorithm}

In order to numerically solve the system of nonlinear equations~\eqref{nonlinear_scheme}, we propose the following iterative scheme for updating the triple $(\vecv^{n}, p^{n}, \mu^n)$ from the time $t_n=n \Delta t$ to the triple $(\vecv^{n+1}, p^{n+1}, \mu^{n+1})$ at the time $t_{n+1}= (n+1) \Delta t$. 

\begin{description}[style=nextline]
\item[\textbf{Initialization:}]
Define $(\vecv^{0}, p^{0}, \mu^0) =_{\bydefinition} (\vecv^{n}, p^{n}, \mu^n)$.

\item[\textbf{Step 1:}]
Given $(\vecv^{\ell}, p^{\ell}, \mu ^\ell)$, to find $(\vec{v}^{\ell+1}, p^{\ell+1})$ such that $\forall\, (\bar{\vecv}, \bar{p}) \in \vec{V}_h \times P_h$:
\begin{equation}
  \label{scheme}
  \begin{array}{r}\displaystyle
    \left(\frac{\vecv^{\ell+1} - \vecv^n}{\Delta t}, \bar{\vecv} \right)
    + \left( 2 \mu^{l} \gradsym^{\ell+1}, \overline{\gradsym} \right) 
    - \left( p^{\ell+1}, \divergence \bar{\vecv} \right)
    = \vec{0},
    \\
    \left( \divergence \vecv^{\ell+1}, \bar{p} \right) = 0.
  \end{array}
\end{equation}

\item[\textbf{Step 2:}]
Compute 
\begin{equation}
  \label{thetasche}
  \mu^{\ell+1} = \frac{1}{2} \left[ \alpha \left( 1 + 4 \beta \absnorm{\mu^\ell}^2 \absnorm{\gradsym^{\ell+1}}^2 \right)^s + \gamma \right]^{-1}.
\end{equation}

\item[\textbf{Step 3:}]
Compute 
\begin{equation*}
  \eta = 
  \norm[L^2(\Omega)]{\mu^{\ell+1} - \mu^{\ell}}
  + \norm[L^2(\Omega)]{\vecv^{\ell+1} - \vecv^{\ell}}
  + \norm[L^2(\Omega)]{{p}^{\ell+1} - {p}^{\ell}}  
\end{equation*}
and then check if
\begin{equation}
  \label{stopcri}
  \left\{
    \begin{array}{rcl}
      \eta > \mathrm{tol} & \Rightarrow 
      & \mbox{update $(\vecv^{\ell}, p^{\ell}, \mu^{\ell}) =_\bydefinition (\vecv^{\ell+1}, p^{\ell+1}, \mu^{\ell+1})$, go to \textbf{Step 1}, and iterate again},
      \\
      \eta \le \mathrm{tol} & \Rightarrow 
      & \mbox{move to the new time step, define } (\vecv^{n+1}, p^{n+1}, \mu^{n+1}) =_{\bydefinition} (\vecv^{\ell+1}, p^{\ell+1}, \mu^{\ell+1}),
    \end{array}
  \right.
\end{equation}
where $\mathrm{tol} > 0$ represents a tolerance parameter/stopping criterion.
\end{description}


\section{Numerical experiments -- reduced problem}
\label{sec:uq}

In order to investigate qualitative features of models based on the implicit non-monotone constitutive relations, we present a reduced version of the system~\eqref{model2}, and introduce a numerical scheme analogous to the scheme proposed in Section~\ref{sec:scheme}. The idea is to design a reduced model that would allow us to see the \emph{qualitative behaviour that is induced by the non-monotone constitutive relation} without the unnecessary complications such as the convective nonlinearity and the incompressibility condition.

\subsection{Reduced model}
\label{sec:model-numer-scheme}
In particular, instead of the vector-tensor variables $(\vecv, \dcstresssymb)$ we consider scalar-vector variables $(u, \vec{q})$, whose evolution is governed by the system
\begin{subequations}
  \label{eq:reduced-problem}
  \begin{align}
    \label{eq:23}
    \pd{u}{t} &= \divergence \vec{q}, \\
    \label{eq:27}
    \nabla u &= \left[ a \left( 1 + b \absnorm{\vec{q}}^2 \right)^n + c \right] \vec{q}.
  \end{align}
\end{subequations}
This system is with respect to the relation between the \emph{flux} $\vec{q}$ and the \emph{affinity} $\nabla u$ structurally similar to~\eqref{model2}, where the \emph{flux} is the Cauchy stress tensor $\traceless{\cstress}$ and the \emph{affinity} is the symmetric part of the velocity gradient $\gradsym$. Note that if we interpret $\vec{q}$ as the heat flux and $u$ as the temperature, then~\eqref{eq:27} corresponds to an implicit variant of Fourier's law, see~\cite{janecka.a.pavelka.m:non-convex} and references therein.

Following the idea exploited in \eqref{def:viscosity}, we define the quantity~$\tilde{\mu}$ as
\begin{equation}
  \tilde{\mu} (\vec{q}) =_\bydefinition \left[ a \left( 1 + b \absnorm{\vec{q}}^2 \right)^n + c \right]^{-1},
\end{equation}
and the problem \eqref{eq:reduced-problem} can be rewritten as a system for $(u, \tilde{\mu})$ as
\begin{subequations}
  \label{eq:reduced-problem-reformulated}
  \begin{align}
      \pd{u}{t} &= \divergence \left( \tilde{\mu} \nabla u \right), \\ 
      \tilde{\mu} &= \left[ a \left(1 + b \tilde{\mu}^2 \absnorm{\nabla u}^2 \right)^n + c \right]^{-1}.
  \end{align}
\end{subequations}
If $n<-\frac{1}{2}$, then the flux--affinity constitutive relation is, in general, non-monotone, and the relation between the norms qualitatively corresponds to that shown in Figure~\ref{fig:constitutive-relation-regions}.
\begin{figure}[h]
  \centering
  \includegraphics[width=0.35\textwidth]{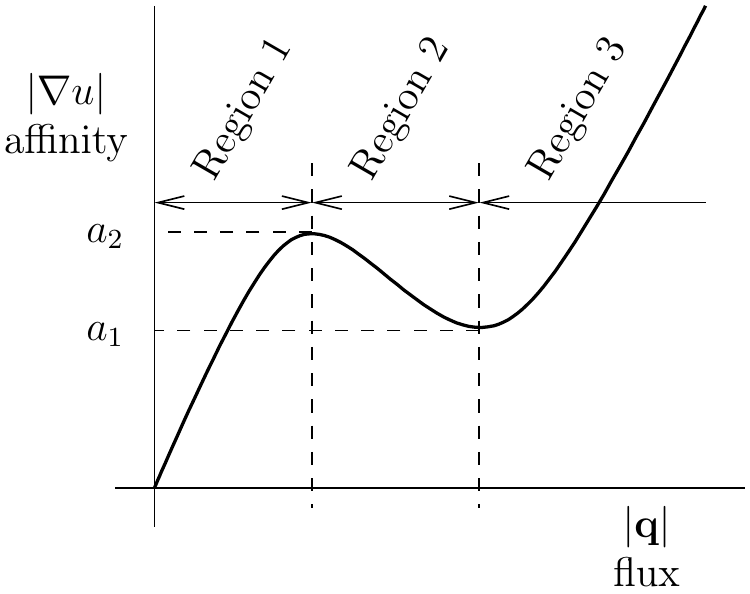}
  \caption{Different regions in the non-monotone constitutive relation.}
  \label{fig:constitutive-relation-regions}
\end{figure}

\subsection{Iterative algorithm}
\label{sec:numerical-scheme-1}

Using the same arguments as in Section~\ref{sec:scheme}, we propose the following iterative scheme for the update of~$u^n$ and $\tilde{\mu}^n$ at the time $t_n$ to $u^{n+1}$ and  $\tilde{\mu}^{n+1}$ at the time $t_{n+1}$.

\begin{description}[style=nextline]
\item[\textbf{Initialization:}]
  Define $\left( u^0, \tilde{\mu}^0 \right) =_{\bydefinition} \left( u^n,  \tilde{\mu}^{n}\right)$.
\item[\textbf{Step 1:}]
  From $(u^{\ell}, \tilde{\mu}^\ell)$, find $u^{\ell+1}$ such that 
  \begin{equation}
    \label{eq:reduced-scheme}
    \left( \frac{u^{\ell+1} - u^n}{\Delta t}, \bar{u} \right)
    + \left(\tilde{\mu}^{\ell} \nabla u^{\ell+1}, \nabla \bar{u} \right) 
    = 0, \qquad
    \forall \bar{u} \in U_h.
  \end{equation}
  
\item[\textbf{Step 2:}]
  Compute 
  \begin{equation}
    \tilde{\mu}^{\ell+1} = \left[ a \left( 1 + b \left( \tilde{\mu}^\ell \right)^2 \absnorm{\nabla u^{\ell+1}}^2 \right)^s + c \right]^{-1}.
  \end{equation}

\item[\textbf{Step 3:}]
  Compute 
  \begin{equation}
    \eta =
    \norm[L^2(\Omega)]{u^{\ell+1} - u^{\ell}}
    + \norm[L^2(\Omega)]{\tilde{\mu}^{\ell+1} - \tilde{\mu}^{\ell}},
  \end{equation}
  and then check if
  \begin{equation}
    \label{eq:reduced-stop-criterium}
    \left\{
      \begin{aligned}
        \eta > \mathrm{tol} \Rightarrow \; &\text{update $(u^\ell, \tilde{\mu}^\ell) =_\bydefinition (u^{\ell+1}, \tilde{\mu}^{\ell+1})$, go to \textbf{Step 1}, and iterate again,} \\
        \eta \leq \mathrm{tol} \Rightarrow \; &\text{move to the new time step, define } (u^n, \tilde{\mu}^n) =_\bydefinition (u^{\ell+1}, \tilde{\mu}^{\ell+1}).
      \end{aligned}
    \right.
  \end{equation}
  where $\mathrm{tol} > 0$ represents a tolerance parameter/stopping criterion.
\end{description}

\subsection{Results}
\label{sec:numerical-results}

Let us present results of several numerical experiments using the numerical scheme \eqref{eq:reduced-scheme}--\eqref{eq:reduced-stop-criterium}. We consider a unit square domain $\Omega =_\bydefinition [0,1]^2$ with $50 \times 50$ triangular mesh. To goal is to determine the behavior of the system depending on the initial and boundary conditions. \emph{Initial conditions} are chosen so that the constitutive relation is satisfied identically in the whole domain with values corresponding to one of the three regions of the non-monotone constitutive relation, see Figure~\ref{fig:constitutive-relation-regions}.

We consider two types of \emph{boundary conditions}. First we consider zero Dirichlet boundary condition $\left. u \right|_{\partial \Omega} = 0$ (Type A) and then the non-homogeneous Dirichlet boundary condition $\left. u \right|_{x = 0} = y (1-y)$, $\left. u \right|_{\partial \Omega \backslash \{x = 0\}} = 0$ (Type B).

We consider four different \emph{initial conditions}. The particular \emph{initial condition} is always specified only by the constant initial vector $\vec{q}^0 = \transpose{\begin{bmatrix} q^0_{\hat{x}} & 0\end{bmatrix}}$, and the initial value of $u$ is given by
\begin{equation}
  \tilde{\mu}^0 u^0 = q^0_{\hat{x}} x,
\end{equation}
with the initial apparent viscosity is given by
\begin{equation}
  \tilde{\mu}^0 = \left[ a \left( 1 + b \absnorm{\vec{q}^0}^2 \right)^n + c \right]^{-1}.
\end{equation}
Various choices of $\vec{q}^0$ always lead to the initial condition for flux--affinity pair that is \emph{consistent with the constitutive relation}~\ref{eq:27}. Using different values of $\vec{q}^0$ one can start with different locations of the initial flux--affinity pair at the constitutive curve, see Figure~\ref{fig:constitutive-relation-regions}.

The problem is solved using a finite element approximation in space and the backward Euler method in time in the FEniCS Project software, see~\cite{logg.a.mardal.k.ea:automated} and \cite{aln-s.m.blechta.j.ea:fenics}. The unknown field~$u$ is approximated by the finite element space $\mathcal{P}_1 =_\bydefinition \left\{ f \in C \left( \overline{\Omega} \right): \left. f \right|_T \in P_1 (T), \forall T \in \mathcal{T}_h \right\}$, whereas the apparent viscosity $\tilde{\mu}$ is approximated by the piecewise constant finite element space $d \mathcal{P}_0 =_\bydefinition \left\{ f \in L^2 \left( \Omega \right): \left. f \right|_T \in P_0 (T), \forall T \in \mathcal{T}_h \right\}$. (The notation $d \mathcal{P}_0$ follows the notation used in~\cite{arnold.dn.logg.a:periodic}.) The idea is to iterate one time step from $t = 0$ to $t = \Delta t$ to understand the dynamics of the constitutive relation. The used parameters are listed in Table~\ref{tab:reduced-problem-parameters}.

\begin{table}[h]
  \centering
  \begin{tabular}{*{7}{c}}
    \toprule
    $\Delta t$ & $a$ & $b$ & $c$ & $n$ & $\mathrm{tol}$  \\
    \midrule
    $10^{-10}$ & $1.0$ & $0.1$ & $10^{-3}$ & $-0.75$ & $10^{-10}$ \\
    \bottomrule
  \end{tabular}
  \caption{Simulation parameters for the reduced problem.}
  \label{tab:reduced-problem-parameters}  
\end{table}

\subsubsection{Case 1: Initial condition in Region 1 and Type B boundary conditions}
\label{sec:case-1:-initial}

Considering $\vec{q}^0 = \transpose{\begin{bmatrix} 3 & 0 \end{bmatrix}}$, all the points are initially in Region 1 of the constitutive curve, see Figure~\ref{fig:case1-ic}. As the time passes, the flux--affinity pairs $[\vec{q}, \nabla u]$ that are evaluated at each mesh point, move along the constitutive curve in such a manner that there is no overlap between Region 1 and Region 3. (Meaning that all the points where the value of the affinity $\nabla u$ allows multiple associated fluxes $\vec{q}$ are located in Region 1. None of the actual flux--affinity pairs is in this presumably ambiguous case located in Region 3.) Moreover, none of the flux--affinity pairs can be found in Region 2 (the decreasing part of the constitutive curve) see Figure~\ref{fig:case1-1it}. The same observation holds true also for later times $t$. 

\begin{figure}[h]
  \centering
  \subfloat[\label{fig:case1-ic} Initial condition. All initial flux--affinity pairs $\vec{q}$, $\nabla u$ are located in Region 1.]{\includegraphics[width=0.45\textwidth]{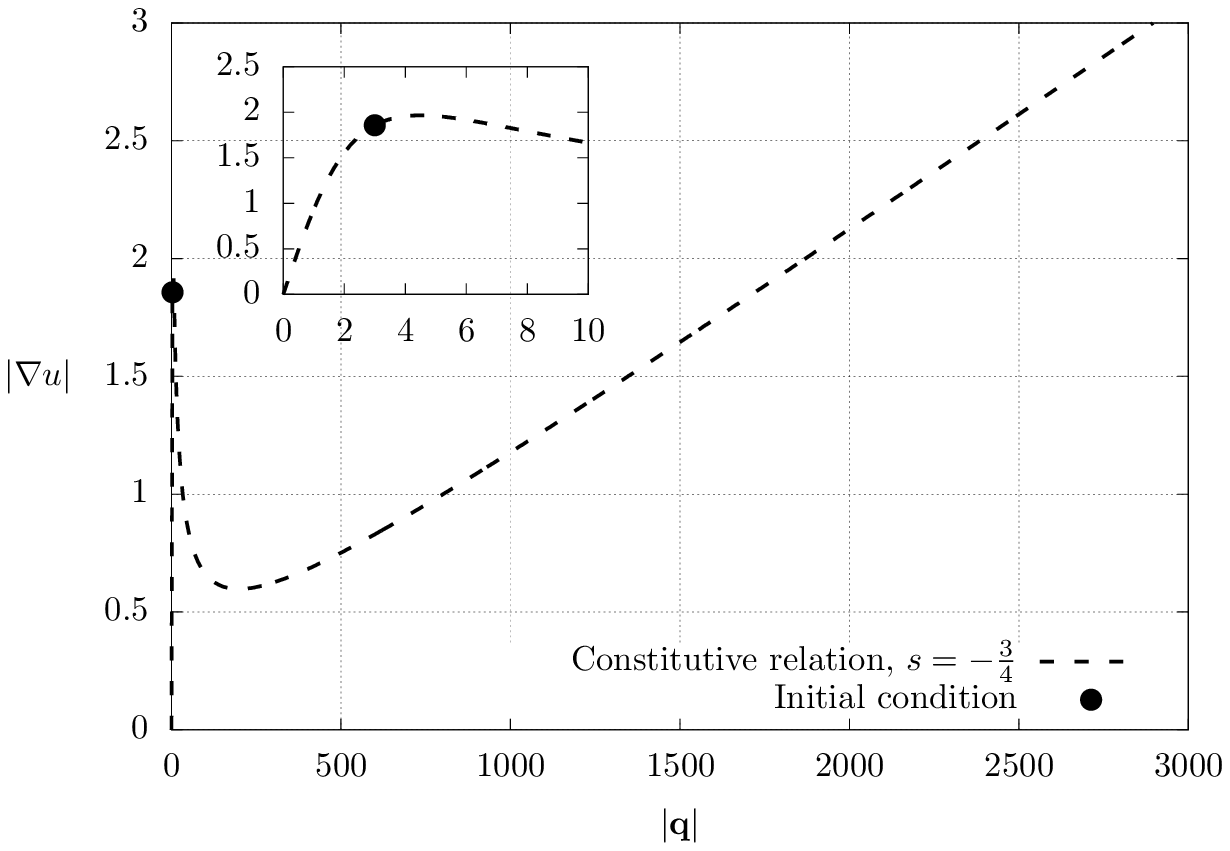}}
  \qquad
  \subfloat[\label{fig:case1-1it} Computed solution at time $t = \Delta t$.]{\includegraphics[width=0.45\textwidth]{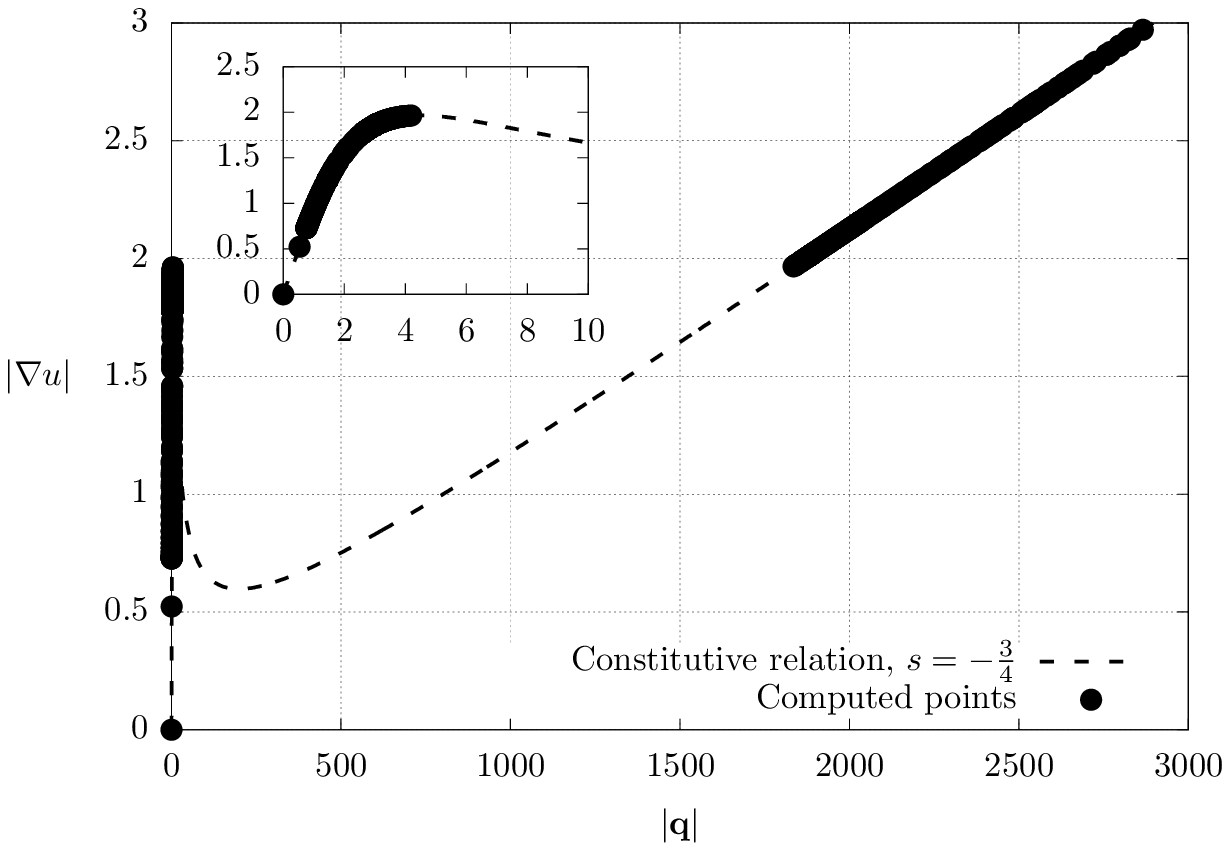}}
  \caption{Case 1. Initial condition and computed solution at $t=\Delta t$, $\vec{q}^0 = \transpose{\left[ 3 \ 0 \right]}$, non-homogeneous Dirichlet boundary condition.}
  \label{fig:case1}
\end{figure}

\subsubsection{Case 2: Initial condition in Region 3 and Type A boundary conditions}
\label{sec:case-2:-initial}

For $\vec{q}^0 = \transpose{\begin{bmatrix} 1000 & 0 \end{bmatrix}}$, all the points are initially in Region 3 of the constitutive curve, see Figure~\ref{fig:case2-ic}. Again, after one time step, and for all later times, there is no overlap between Region~1 and Region~3, and no actual flux--affinity pair is located in Region 2 of the constitutive curve, see Figure~\ref{fig:case2-1it}. In order to resolve all flux-affinity pairs for small values of $\absnorm{\vec{q}}$, we had to use eight times denser mesh than in Case 1.

\begin{figure}[h]
  \centering
  \subfloat[\label{fig:case2-ic} Initial condition. All initial flux--affinity pairs $\vec{q}$, $\nabla u$ are located in Region 3.]{\includegraphics[width=0.45\textwidth]{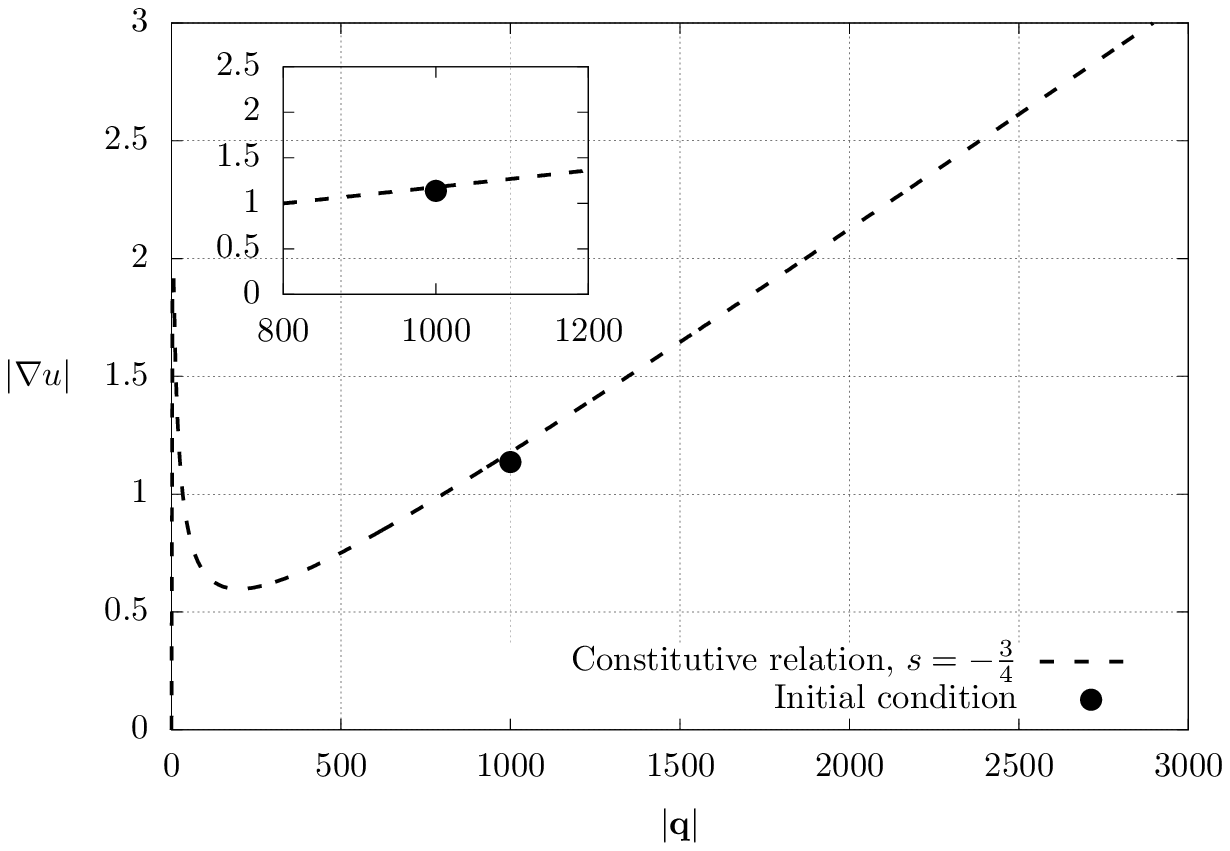}}
  \qquad
  \subfloat[\label{fig:case2-1it} Computed solution at time $t = \Delta t$.]{\includegraphics[width=0.45\textwidth]{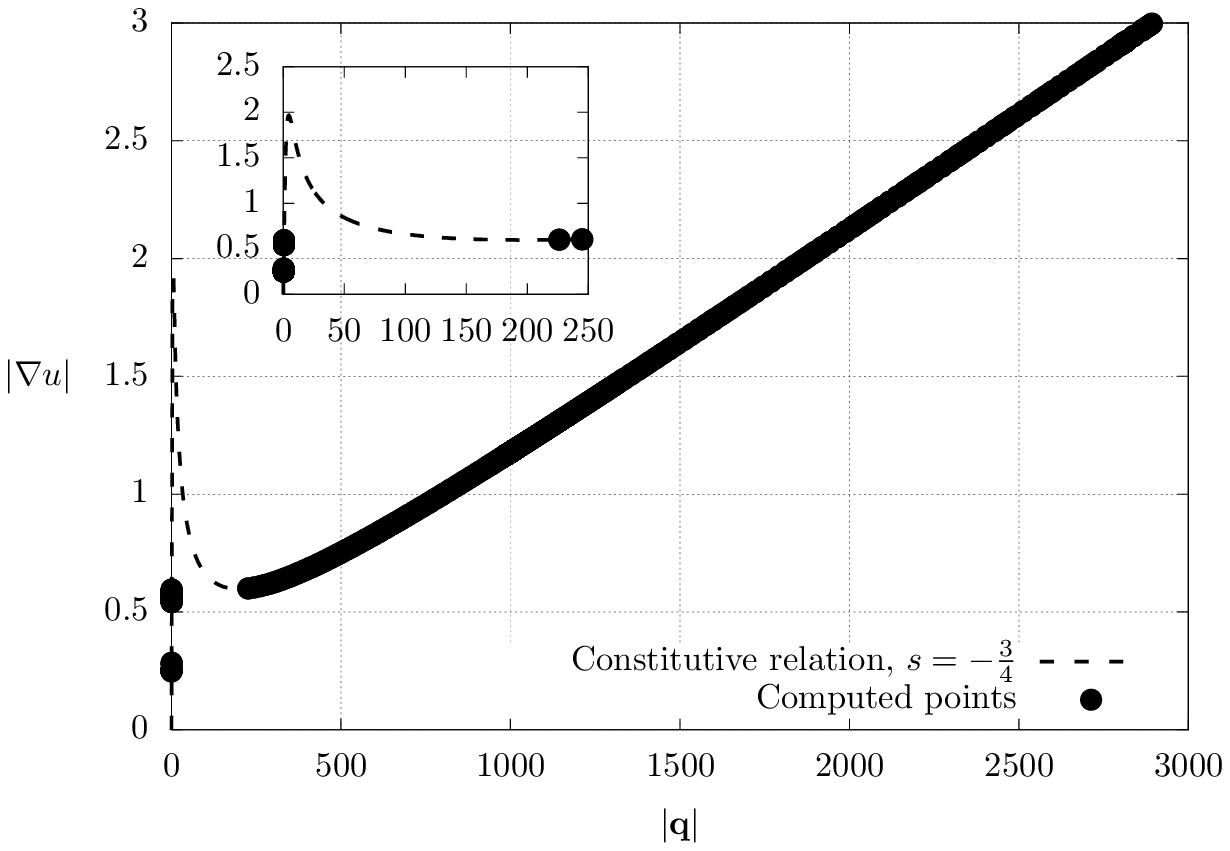}}
  \caption{Case 2. Initial condition and computed solution at $t=\Delta t$, $\vec{q}^0 = \transpose{\left[ 1000 \ 0 \right]}$, homogeneous Dirichlet boundary condition.}
  \label{fig:case2}
\end{figure}

\subsubsection{Case 3: Initial condition in Region 2 and Type A boundary conditions}
\label{sec:case-3:-initial}

Here, $\vec{q}^0 = \transpose{\begin{bmatrix} 25 & 0 \end{bmatrix}}$, hence all the flux-affinity pairs are initially located in Region 2, see Figure~\ref{fig:case3-ic}. As the time evolves, the flux--affinity pairs move from Region 2 to Region 1 and Region 3. Again there is no overlap between these two regions, see Figure~\ref{fig:case3-1it}.

\begin{figure}[h]
  \centering
  \subfloat[\label{fig:case3-ic} Initial condition. All initial flux--affinity pairs $\vec{q}$, $\nabla u$ are located in Region 2.]{\includegraphics[width=0.45\textwidth]{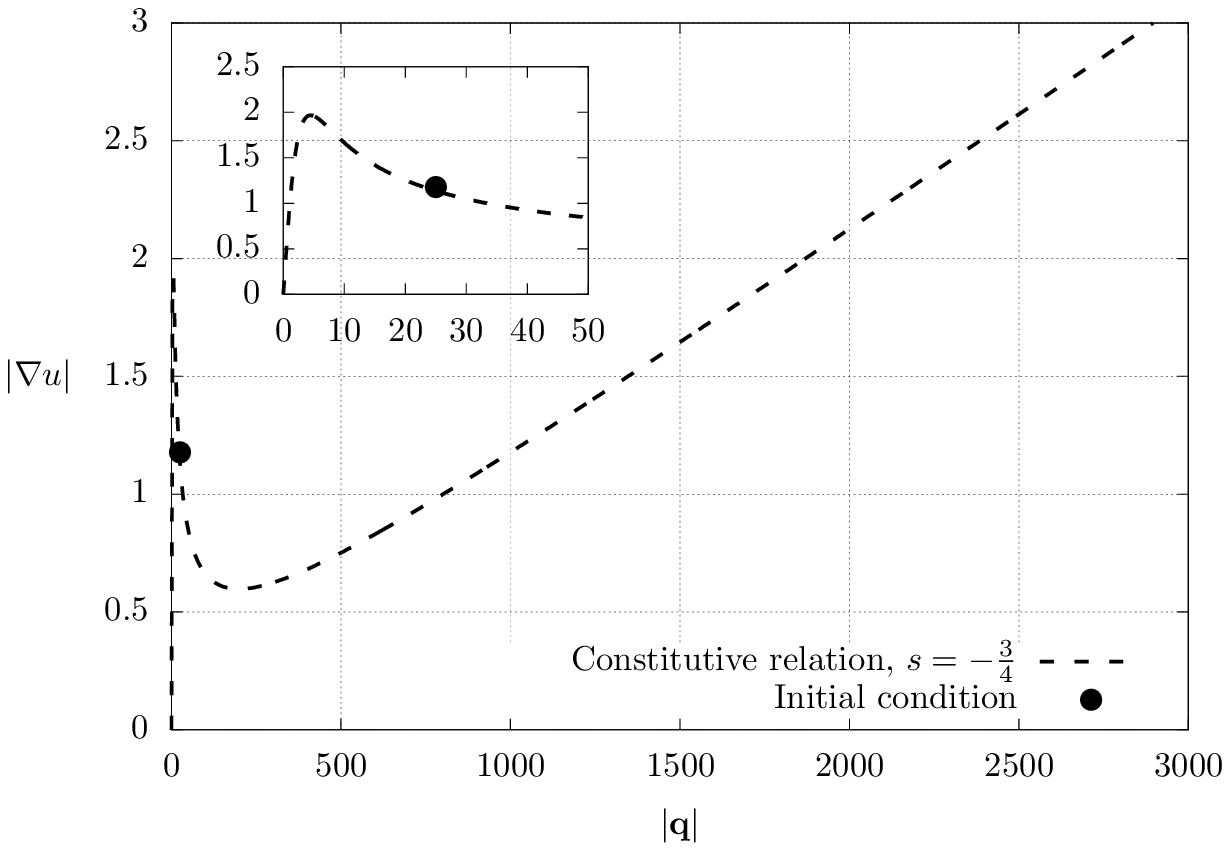}}
  \qquad
  \subfloat[\label{fig:case3-1it} Computed solution at time $t = \Delta t$.]{\includegraphics[width=0.45\textwidth]{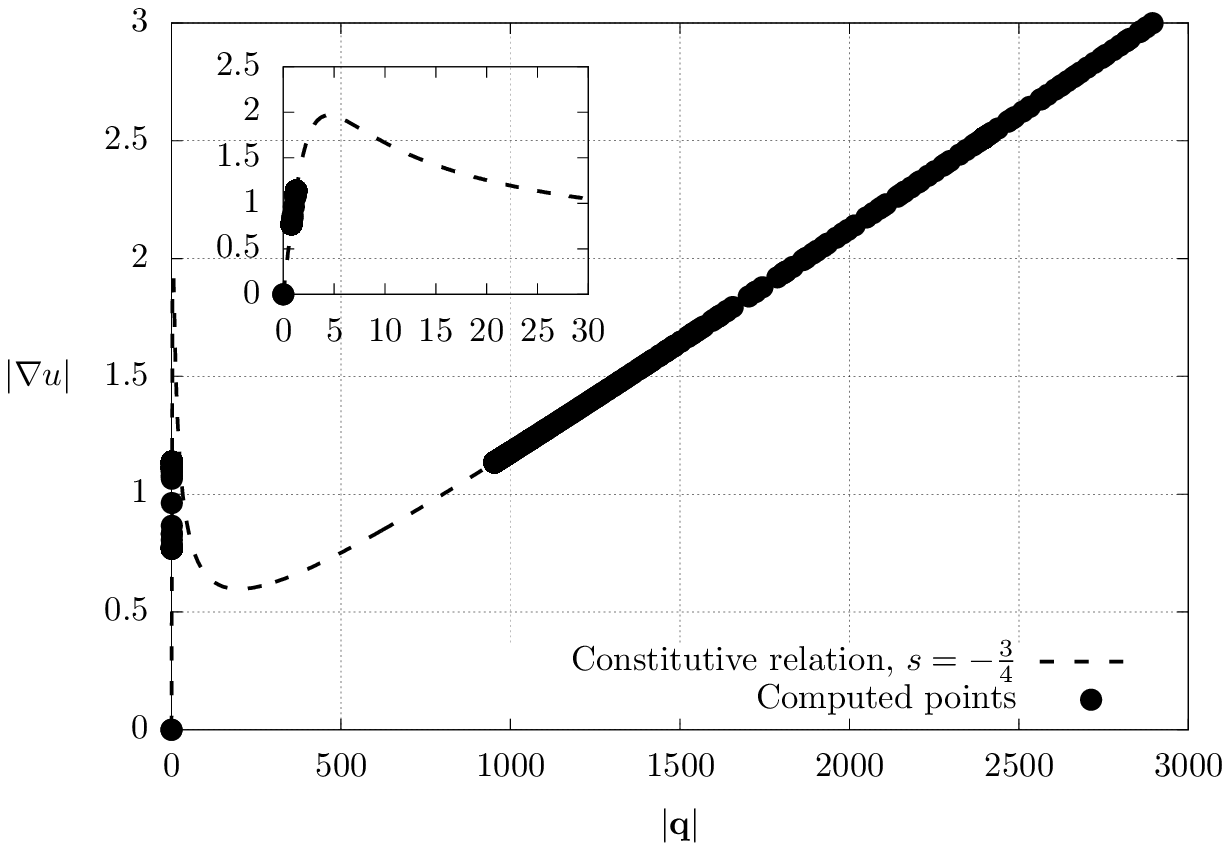}}
  \caption{Case 3. Initial condition and computed solution at $t=\Delta t$, $\vec{q}^0 = \transpose{\left[ 25 \ 0 \right]}$, homogeneous Dirichlet boundary condition.}
  \label{fig:case3}
\end{figure}

\subsubsection{Case 4: Initial condition in Region 3 and Type B boundary conditions}
\label{sec:case-4:-initial}

Now, the initial condition is the same as in Case 3, that is $\vec{q}^0 = \transpose{\begin{bmatrix} 25 & 0 \end{bmatrix}}$, hence all the flux-affinity pairs are again initially located in Region 2, see Figure~\ref{fig:case4-ic}. On the other hand, the boundary condition is now the non-homogeneous Dirichlet boundary condition. The qualitative behavior is however identical to Case 3, while the only difference is higher number of points in Region~1, see Figure~\ref{fig:case4-1it}.

\begin{figure}[h]
  \centering
  \subfloat[\label{fig:case4-ic} Initial condition. All initial flux--affinity pairs $\vec{q}$, $\nabla u$ are located in Region 2.]{\includegraphics[width=0.45\textwidth]{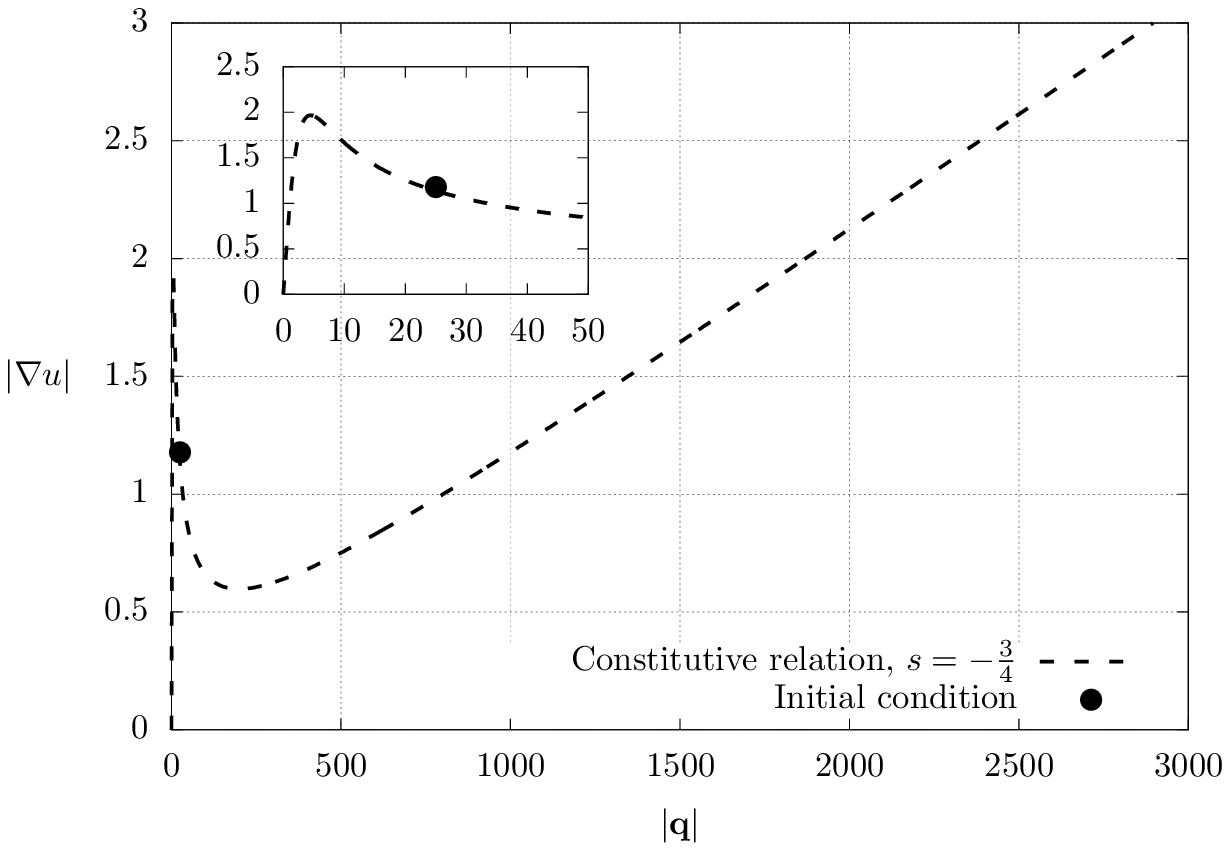}}
  \qquad
  \subfloat[\label{fig:case4-1it} Computed solution at time $t = \Delta t$.]{\includegraphics[width=0.45\textwidth]{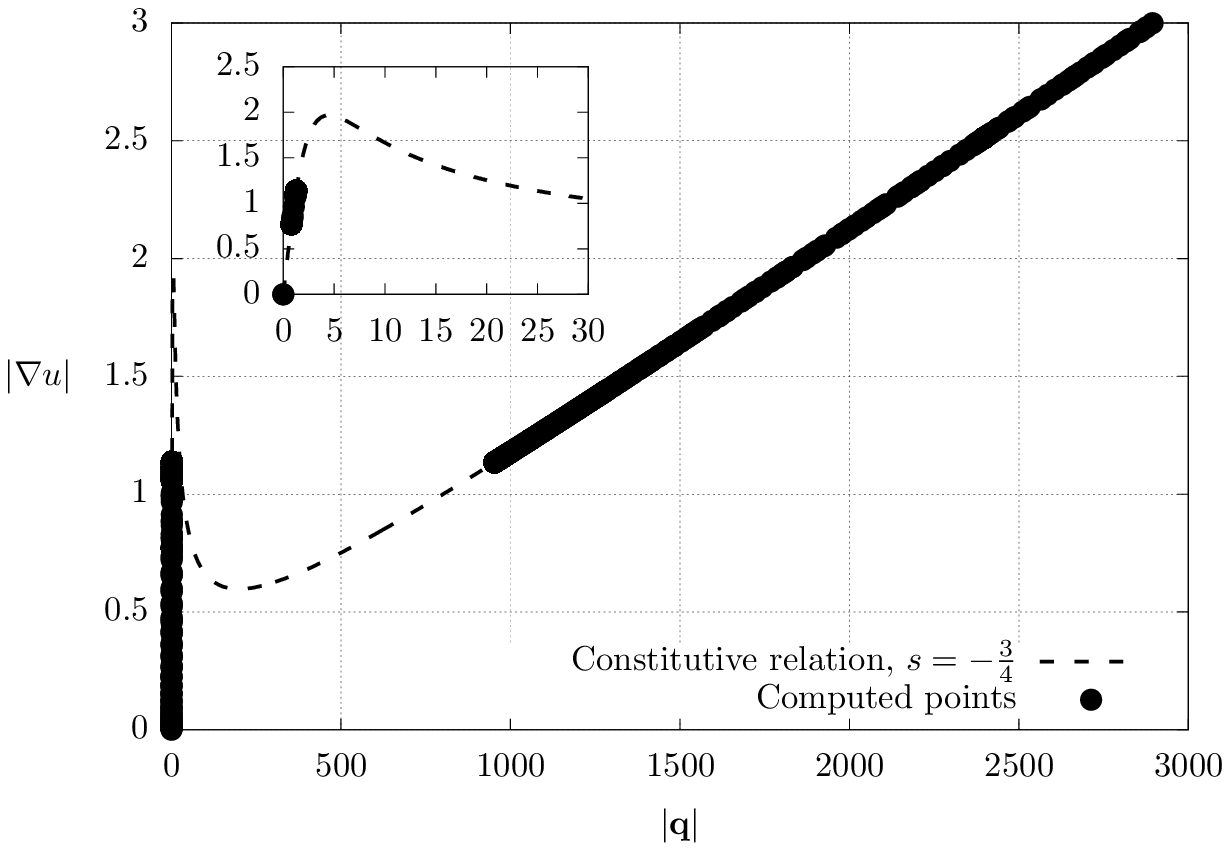}}
  \caption{Case 4. Initial condition and computed solution at $t=\Delta t$, $\vec{q}^0 = \transpose{\left[ 25 \ 0 \right]}$, non-homogeneous Dirichlet boundary condition.}
  \label{fig:case4}
\end{figure}

\subsubsection{Comments on numerical experiments}
\label{sec:conclusion-1}

We have designed simple numerical experiments that allowed us to investigate quantitative and qualitative behaviour of a system whose response is described by an implicit constitutive relation~\eqref{eq:27}). The chosen constitutive relation predicts, see Figure~\ref{fig:constitutive-relation-regions}, that once the affinity $\nabla u$ reaches the range $(a_1, a_2)$, then there exist several fluxes $\vec{q}$ such that the corresponding flux--affinity pair lies on the constitutive curve. This behaviour qualitatively corresponds to the behaviour of relation between the flux (Cauchy stress tensor, $\traceless{\cstress}$) and the affinity (symmetric part of the velocity gradient, $\gradsym$) in the case of more complex constitutive relation~\eqref{eq:9}. Apparently, such a behaviour should lead to ambiguous specification of actual flux--affinity pairs.

The numerical experiments however indicate that once the problem is solved as an evolution problem, then no ambiguity arises. The position of actual flux--affinity pairs is fully determined by the initial conditions, boundary conditions and the evolution equation for the linear momentum. In particular, it seems that no actual flux--affinity pair can over time occupy Region 2, which corresponds to unstable flux--affinity pairs. This is in agreement with the thermodynamical stability analysis given in~\cite{janecka.a.pavelka.m:non-convex}.


\section{Numerical experiments -- full problem}
\label{sec:simulations}
Using the proposed numerical scheme, we finally solve various initial--boundary value problems for the fluid described by the non-monotone implicit constitutive relation~\eqref{eq:4}. The proposed numerical scheme has been implemented in FreeFem++ software, see~\cite{freefem++hecht}, as well as in FEniCS Project software, see~\cite{logg.a.mardal.k.ea:automated} and \cite{aln-s.m.blechta.j.ea:fenics}, that are general purpose software packages for solving partial differential equations using the finite element method.

In FreeFem++ the pressure--velocity pair $(p, \vecv)$ has been approximated by the mini-element $\mathcal{P}_1 \times \mathcal{P}_1-bubble$. In FEniCS the pressure-velocity pair $(p, \vecv)$ has been approximated by the standard lowest order Taylor--Hood elements $\left( \mathcal{P}_1, \vec{\mathcal{P}}_2 \right)$, where $\mathcal{P}_k =_\bydefinition \left\{ v \in C \left( \overline{\Omega} \right): \left. v \right|_T \in P_k (T), \forall T \in \mathcal{T}_h \right\}$ is the Lagrange element of order $k$ and $\vec{\mathcal{P}}_k$ is its vectorial counterpart. For the viscosity, it is not clear how to choose the appropriate finite element function space. Since it is computed as a function of the discontinuous velocity gradient from \eqref{thetasche}, we have used, both in FEniCS and FreeFem++, the piecewise constant approximation $d \mathcal{P}_0$ as the lowest order discontinuous Lagrange element $d \mathcal{P}_k =_\bydefinition \left\{ v \in L^2 \left( \Omega \right): \left. v \right|_T \in P_k (T), \forall T \in \mathcal{T}_h \right\}$. For the temporal discretization, we have used the Crank--Nicolson method. The experimental error analysis of the proposed numerical scheme is presented elsewhere, see~\cite{malek.j.tierra.g:numerical}.

\subsection{Cylindrical Couette flow}
\label{sec:cylindrical-couette}

First, we study the behavior of a fluid described by the non-monotone constitutive relation \eqref{eq:4} in the cylindrical Couette setting. This setting provides a two-dimensional simplification of the typical experimental setting used in rheology, see for example~\cite{donnelly.rj:taylor-couette} for a historical review. In the cylindrical Couette flow problem, the fluid under investigation is confined in between two infinite concentric cylinders $\Gamma_1$ and $\Gamma_2$ of radii $R_1$ and $R_2$ respectively, $R_1 < R_2$, see Figure~\ref{fig:place-couette}, and the flow is induced by the rotation of the cylinders. 

\begin{figure}[ht]
  \centering
  \includegraphics[width=0.25\textwidth]{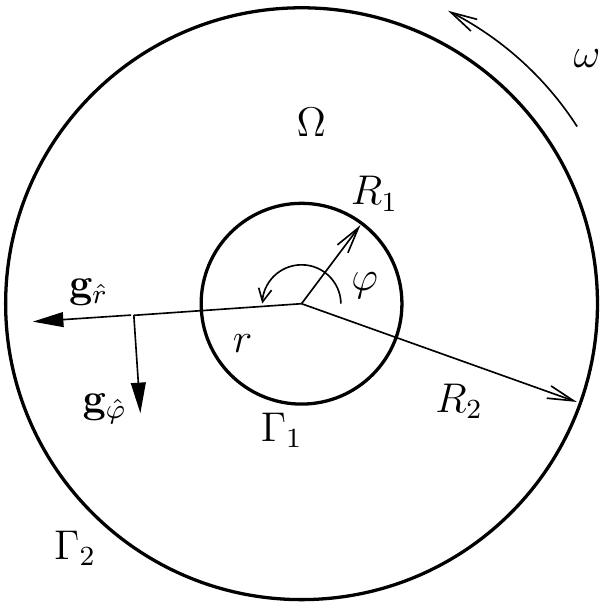}
  \caption{Cylindrical Couette flow -- problem geometry.}
  \label{fig:place-couette}
\end{figure}

In particular, we are interested in the setting where the inner cylinder is at rest and the outer cylinder rotates with a prescribed angular velocity $\omega$. This corresponds to the so-called \emph{shear-rate controlled} experiment. In this experiment, one controls the shear-rate through the control of the angular velocity\footnote{Indeed, if the gap between the cylinders is relatively small, then the shear-rate can be well approximated by the velocity difference between the cylinders, that is $\absnorm{\gradsym} \sim \frac{\vhatp(R_2) - \vhatp(R_1)}{R_2 - R_1}$.} $\omega$, and one measures the torque $\tau$ exerted by the flowing fluid on the outer cylinder. 

If the inner cylinder is at rest, then the corresponding boundary condition on the inner cylinder reads
\begin{subequations}
  \label{eq:31}
  \begin{equation}
    \label{eq:29}
    \left. \vecv \right|_{r=R_1} = \vec{0}.
  \end{equation}
  Further, if the outer one rotates with a prescribed time-dependent angular velocity $\omega$, then the velocity on the boundary is $V(t) = \omega(t) R_2$, and the corresponding boundary condition on the outer cylinder reads
  \begin{equation}
    \label{eq:28}
    \left. \vecv \right|_{r=R_2} = \omega (t) R_2 \cobvecn{\varphi},
  \end{equation}
\end{subequations}
where $\cobvecn{\varphi}$ is the azimuthal base vector in the cylindrical coordinate system, see Figure~\ref{fig:place-couette}. The second boundary condition can be further expressed in the Cartesian coordinate system as
\begin{equation}
  \label{eq:30}
  \left. \vecv \right|_{r=R_2} = \omega(t) \left( -y \bvecx + x \bvecy \right), 
\end{equation}
where~$\bvecx$ and $\bvecy$ denote the Cartesian base vectors. 

Once the velocity field is found as a solution to~\eqref{eq:5}, the torque $\tau$ acting on the outer cylinder is found using the formula
\begin{equation}
  \label{eq:32}
    \tau =_\bydefinition \int_{\Gamma_2} R_2 \cobvecn{r} \times \cstress \cobvecn{r} \,\diff l = R_2 \int_{\Gamma_2} \cobvecn{r} \times \left[ \left( \cobvecn{r} \cdot \cstress \cobvecn{r} \right) \cobvecn{r} + \left( \cobvecn{\varphi} \cdot \cstress \cobvecn{r} \right) \cobvecn{\varphi} \right] \,\diff l 
    =
    \left( R_2 \int_{\Gamma_2} \cstressc_{\hat{\varphi} \hat{r}} \,\diff l \right) \cobvecn{z},
\end{equation}
where $\{ \cobvecn{r}, \cobvecn{\varphi}, \cobvecn{z} \}$ denotes the basis in the cylindrical coordinate system, $\cstressc_{\hat{\varphi} \hat{r}}$ is the relevant component of the Cauchy stress tensor $\cstress$ and $\diff l$ is the line element. Again, we can express the torque in the Cartesian coordinate system as
\begin{equation}
  \label{eq:33}
  \tau = \left\{ \frac{1}{R_2} \int_{\Gamma_2} \left[ \left( \cstressc_{\haty\haty} - \cstressc_{\hatx\hatx} \right) xy + \cstressc_{\hatx\haty} (x^2 - y^2) \right] \,\diff l \right\} \bvecz.
\end{equation}

Concerning the angular velocity of the outer cylinder, we consider time-dependent angular velocity~$\omega(t)$ in the form
\begin{equation}
  \label{eq:imposed-omega}
  \omega(t) =
  \begin{cases}
    \omega_0 \sin \left( \pi \frac{t}{t_0} \right), & t \leq t_0, \\
    0, & t > t_0,
  \end{cases}
\end{equation}
with $\omega_0 = 0.2$ and $t_0 = 2 \times 10^{-8}$, see Figure~\ref{fig:angular-velocity}.

\begin{figure}[ht]
  \centering
  \centering
  \subfloat[\label{fig:angular-velocity} Imposed angular velocity $\omega$.]{\includegraphics[width=0.45\textwidth]{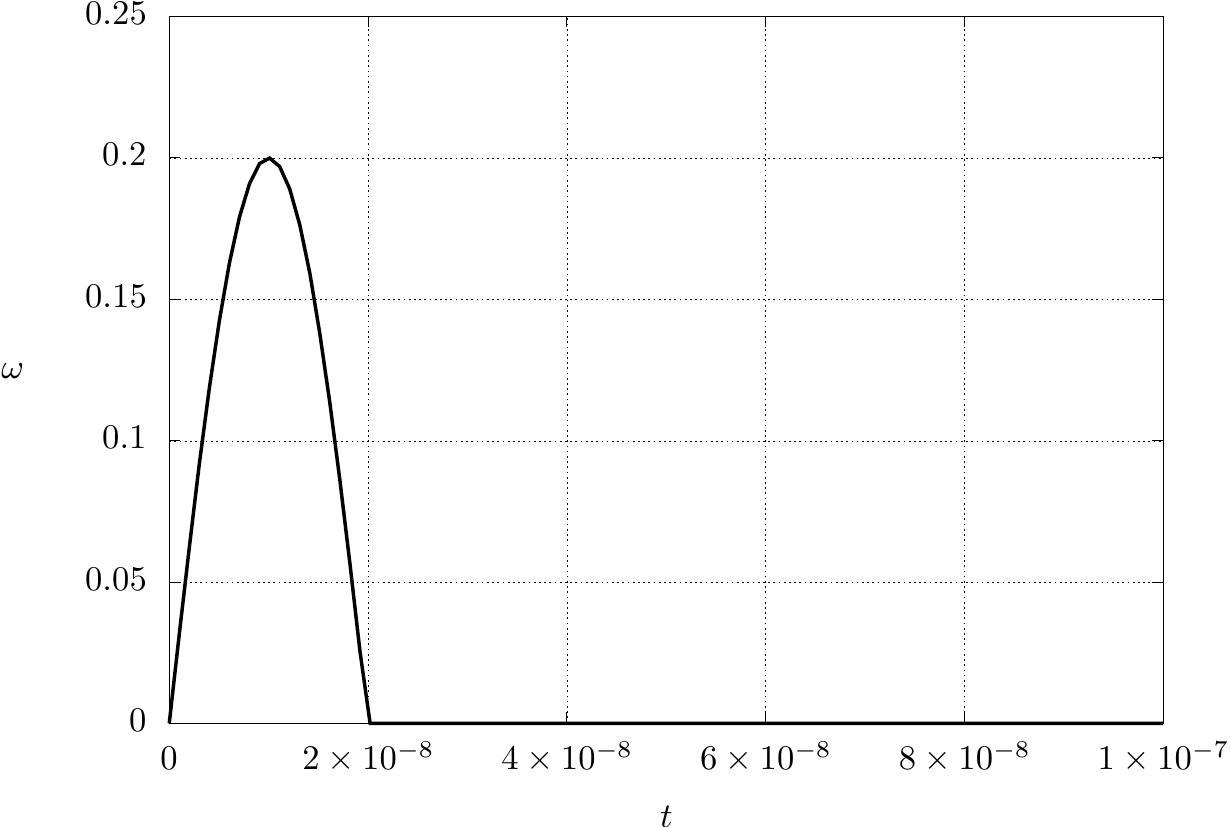}}
  \qquad
  \subfloat[\label{fig:computed-torque} Computed torque $\tau$.]{\includegraphics[width=0.45\textwidth]{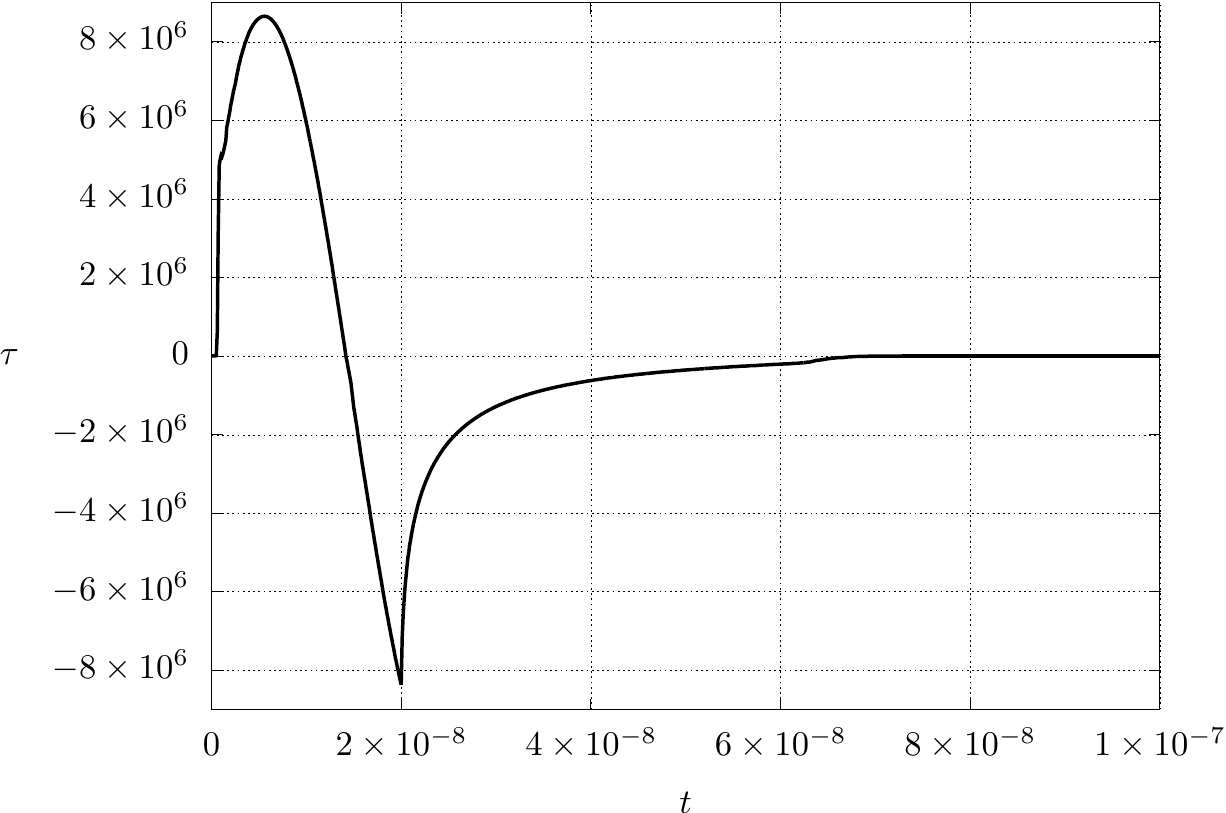}}
  \caption{Shear-rate controlled experiment. Imposed angular velocity $\omega$ versus the torque $\tau$ evaluated using~\eqref{eq:33} and the computed velocity filed.}
  \label{fig:shear-rate-controlled-results}
\end{figure}

Note that the maximal value of the angular velocity $\Omega$ is chosen in such a way that the shear-rate is expected, in certain time interval, to enter the region where the S-shaped constitutive curve, see Figure~\ref{fig:sshaped}, formally allows multiple flux--affinity (stress--shear-rate) pairs. Other material and geometrical parameters used in the numerical simulations are listed in Table~\ref{tab:cylindrical-couette-parameters}. The spatial discretisation of the computational domain contained 16984 cells with the minimum cell size $0.0125$, and maximum cell size $0.031$. Total number of degrees of freedom (DOF) for the unknown fields was $\mathrm{DOF}_{\text{velocity}} = 34356$, $\mathrm{DOF}_{\text{pressure}} = 8686$, $\mathrm{DOF}_{\text{viscosity}} = 16984$.

\begin{table}[ht]
  \centering
  \begin{tabular}{*{9}{c}}
    \toprule
    $R_1$ & $R_2$ & $\alpha$ & $\beta$ & $\gamma$ & $s$ & $\Delta t$ 
    & $\mathrm{tol}$  \\
    \midrule
    $0.3$ & $1.0$ & $1.0$ & $0.1$ & $10^{-6}$ & $-0.75$ & $10^{-10}$ 
    & $10^{-12}$  \\
    \bottomrule
  \end{tabular}
  \caption{Parameters used in the numerical experiments in the cylindrical Couette flow problem.}
  \label{tab:cylindrical-couette-parameters}  
\end{table}

The computed velocity field $\vec{v}$ and the apparent viscosity filed $\tilde{\mu}$ that correspond to the forcing induced by the imposed angular velocity $\omega$ are shown in Figure~\ref{fig:couette-velocity} and Figure~\ref{fig:couette-viscosity}. We see that the initially quiescent fluid starts to move as the angular velocity of the outer cylinder increases. The flow takes place in a thin layer close to the outer cylinder, where the apparent viscosity $\widetilde{\mu}$ is high, see Figure~\ref{fig:couette-viscosity}, and where the flux--affinity pairs, now given by $[\traceless{\cstress}, \gradsym]$, belong to Region 3 on the constitutive curve. In the remaining part of the flow domain, the flux--affinity pairs occupy Region 1 on the constitutive curve. (See Figure~\ref{fig:constitutive-relation-regions} for the notation concerning various regions on the constitutive curve.) However, the interface between low viscosity and high viscosity regions is blurry and its \emph{exact} position seems to depend on the tolerances used in the numerical method and on the quality of the mesh in the interfacial region. On the other hand, the overall ``averaged'' position of the interface seems to be quite robust with respect to the choice of parameters in the numerical method. The same also holds for the computed torque $\tau$.

Further, Figure~\ref{fig:SR_20} documents that the computed flux--affinity pairs indeed lie on the constitutive curve, and that that flux--affinity pairs never lie in Region 2 on the constitutive curve. (Recall that Region 2 corresponds to unstable flux--affinity pairs.) This is again in agreement with the thermodynamical stability analysis given in~\cite{janecka.a.pavelka.m:non-convex}.

Finally, we also plot the torque~$\tau$ acting on the outer cylinder, see Figure~\ref{fig:computed-torque}. The peak values of the torque are slightly delayed with respect to the peak values of the angular velocity, and as the angular velocity vanishes the torque also finally recovers the zero value as expected.

\begin{figure}[h]
  \centering
  \captionsetup[subfigure]{labelformat=empty}
  \subfloat[$t = 10^{-10}$]{%
    \includegraphics[width=0.18\textwidth]{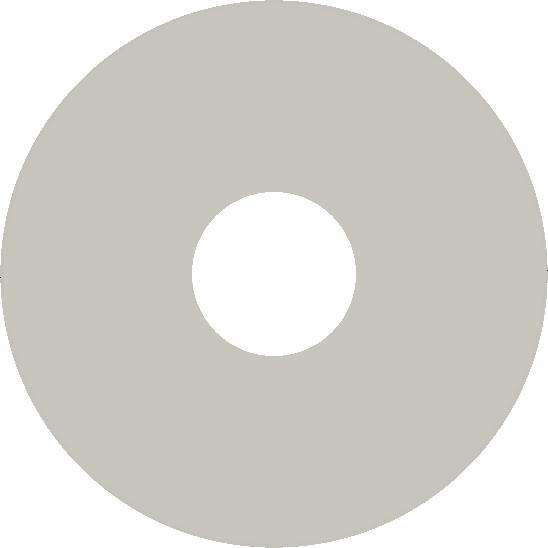}}
  \qquad
  \subfloat[$t = 10^{-8}$]{%
    \includegraphics[width=0.18\textwidth]{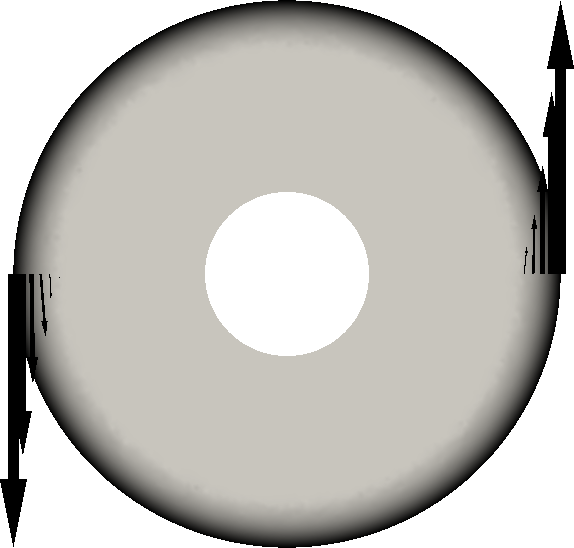}}
  \qquad
  \subfloat[$t = 2 \times 10^{-8}$]{%
    \includegraphics[width=0.18\textwidth]{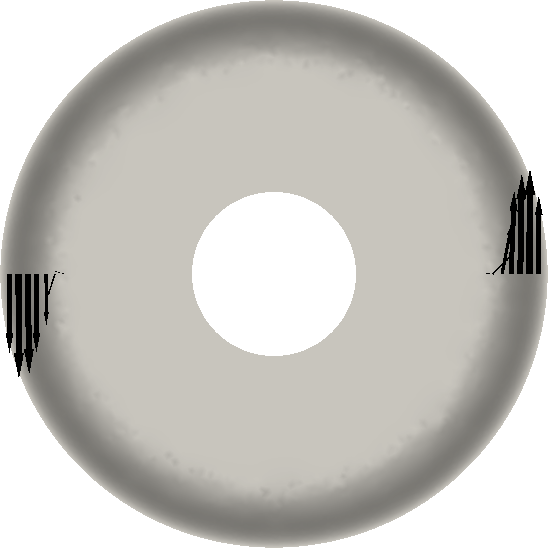}}
  \qquad
  \subfloat[$t = 3 \times 10^{-8}$]{%
    \includegraphics[width=0.18\textwidth]{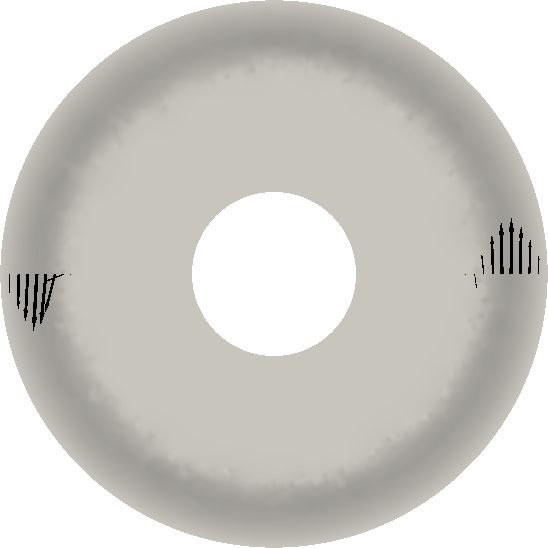}}
  \\
  \subfloat[$t = 4 \times 10^{-8}$]{%
    \includegraphics[width=0.18\textwidth]{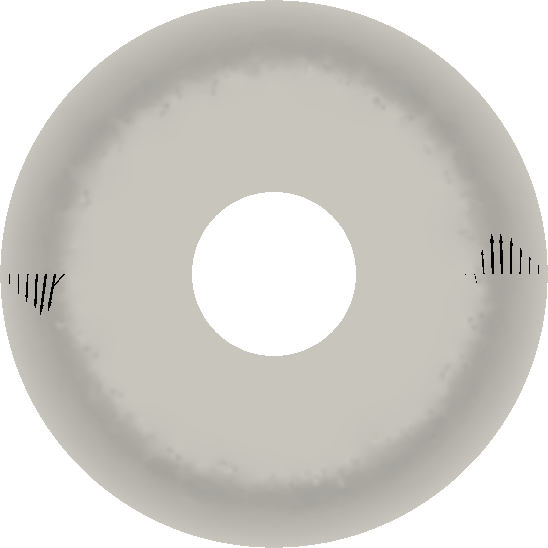}}
  \qquad
  \subfloat[$t = 6.0 \times 10^{-8}$]{%
    \includegraphics[width=0.18\textwidth]{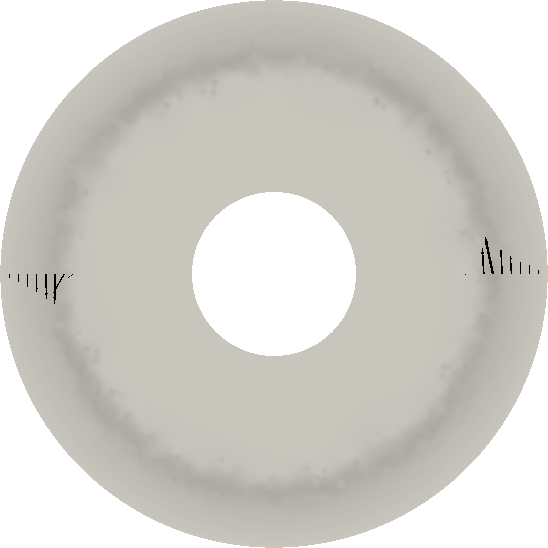}}
  \qquad
  \subfloat[$t = 6.5 \times 10^{-8}$]{%
    \includegraphics[width=0.18\textwidth]{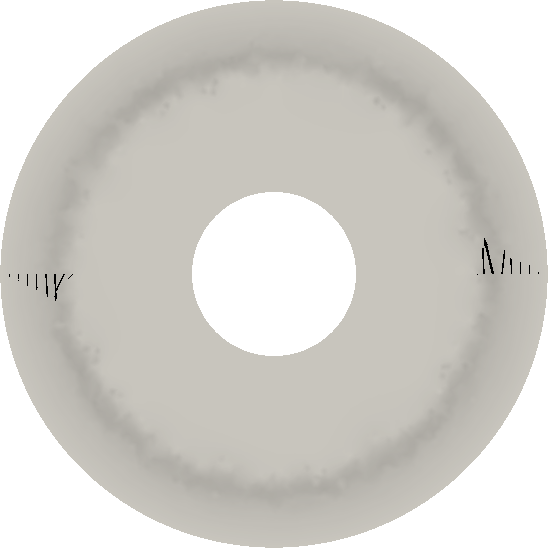}}
  \qquad
  \subfloat[$t = 7 \times 10^{-8}$]{%
    \includegraphics[width=0.18\textwidth]{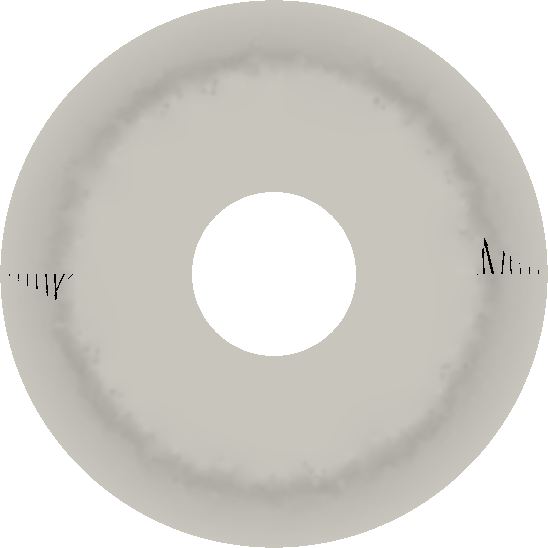}}
  \\
  \subfloat{%
    \includegraphics[width=0.18\textwidth]{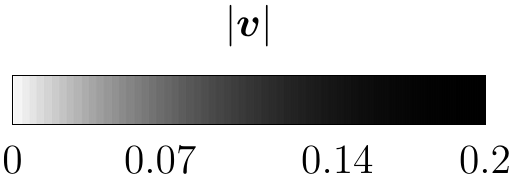}}
  \caption{Computed time evolution of the velocity field~$\vec{v}$ in the cylindrical Couette flow driven by the the imposed angular velocity $\omega$, see~\eqref{eq:imposed-omega} and~Figure~\ref{fig:angular-velocity}.}
  \label{fig:couette-velocity}
\end{figure}

\begin{figure}[h]
  \centering
  \captionsetup[subfigure]{labelformat=empty}
  \subfloat[$t = 10^{-10}$]{%
    \includegraphics[width=0.18\textwidth]{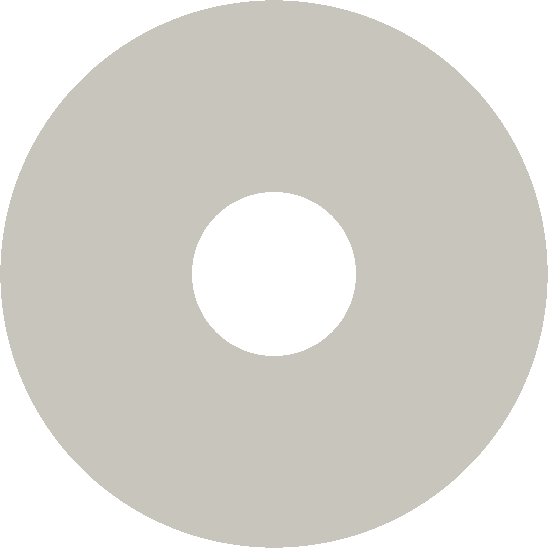}}
  \qquad
  \subfloat[$t = 10^{-8}$]{%
    \includegraphics[width=0.18\textwidth]{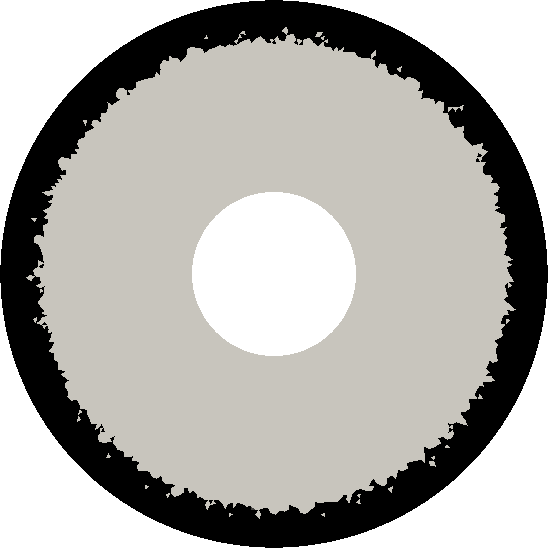}}
  \qquad
  \subfloat[$t = 2 \times 10^{-8}$]{%
    \includegraphics[width=0.18\textwidth]{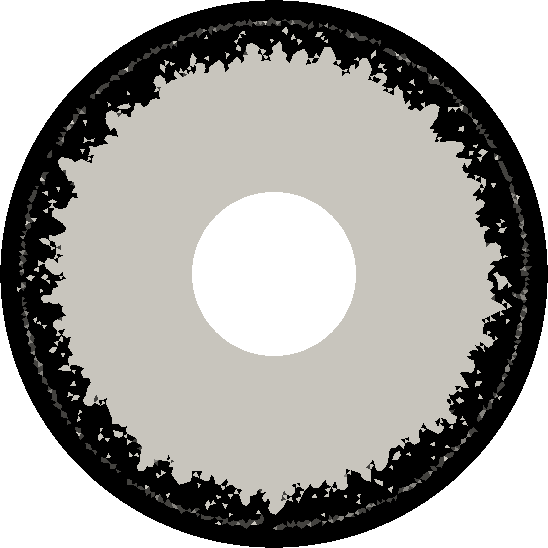}}
  \qquad
  \subfloat[$t = 3 \times 10^{-8}$]{%
    \includegraphics[width=0.18\textwidth]{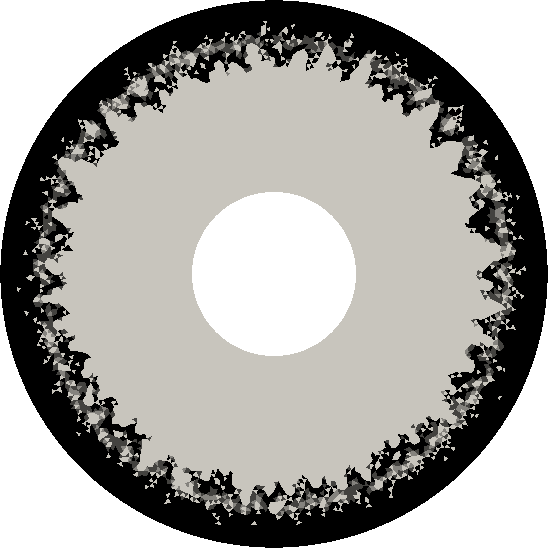}}
  \\
  \subfloat[$t = 4 \times 10^{-8}$]{%
    \includegraphics[width=0.18\textwidth]{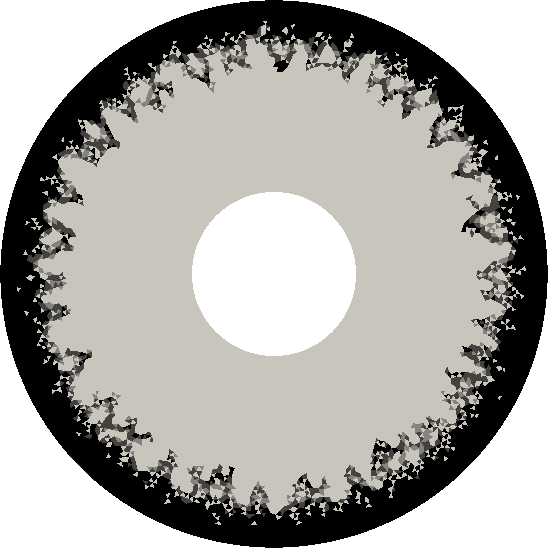}}
  \qquad
  \subfloat[$t = 6.0 \times 10^{-8}$]{%
    \includegraphics[width=0.18\textwidth]{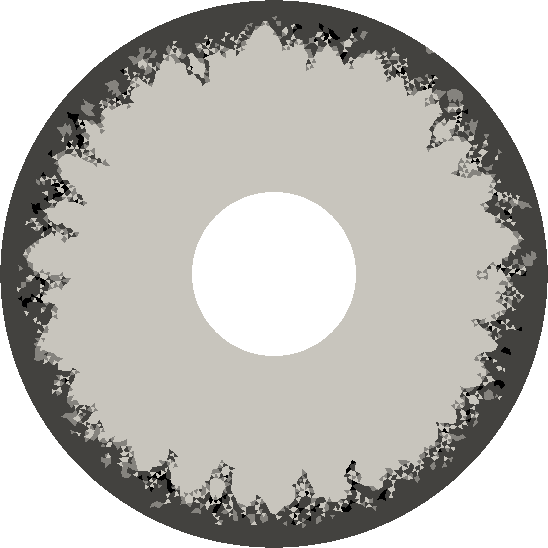}}
  \qquad
  \subfloat[$t = 6.5 \times 10^{-8}$]{%
    \includegraphics[width=0.18\textwidth]{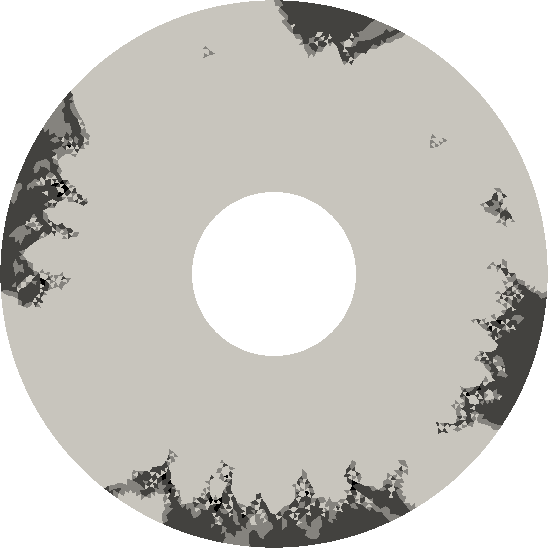}}
  \qquad
  \subfloat[$t = 7 \times 10^{-8}$]{%
    \includegraphics[width=0.18\textwidth]{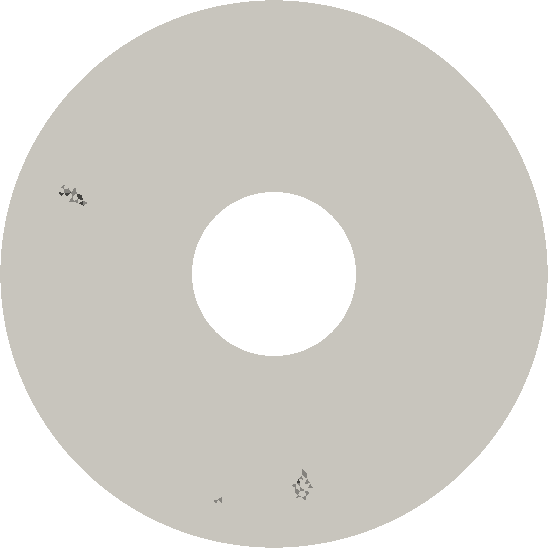}}
  \\
  \subfloat{%
    \includegraphics[width=0.18\textwidth]{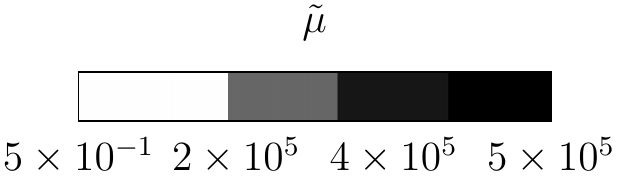}}
  \caption{Computed time evolution of the apparent viscosity $\tilde{\mu}$ in the cylindrical Couette flow driven by the the imposed angular velocity $\omega$, see~\eqref{eq:imposed-omega} and~Figure~\ref{fig:angular-velocity}.}
  \label{fig:couette-viscosity}
\end{figure}

\begin{figure}[h]
  \centering
  \includegraphics[width=0.45\textwidth]{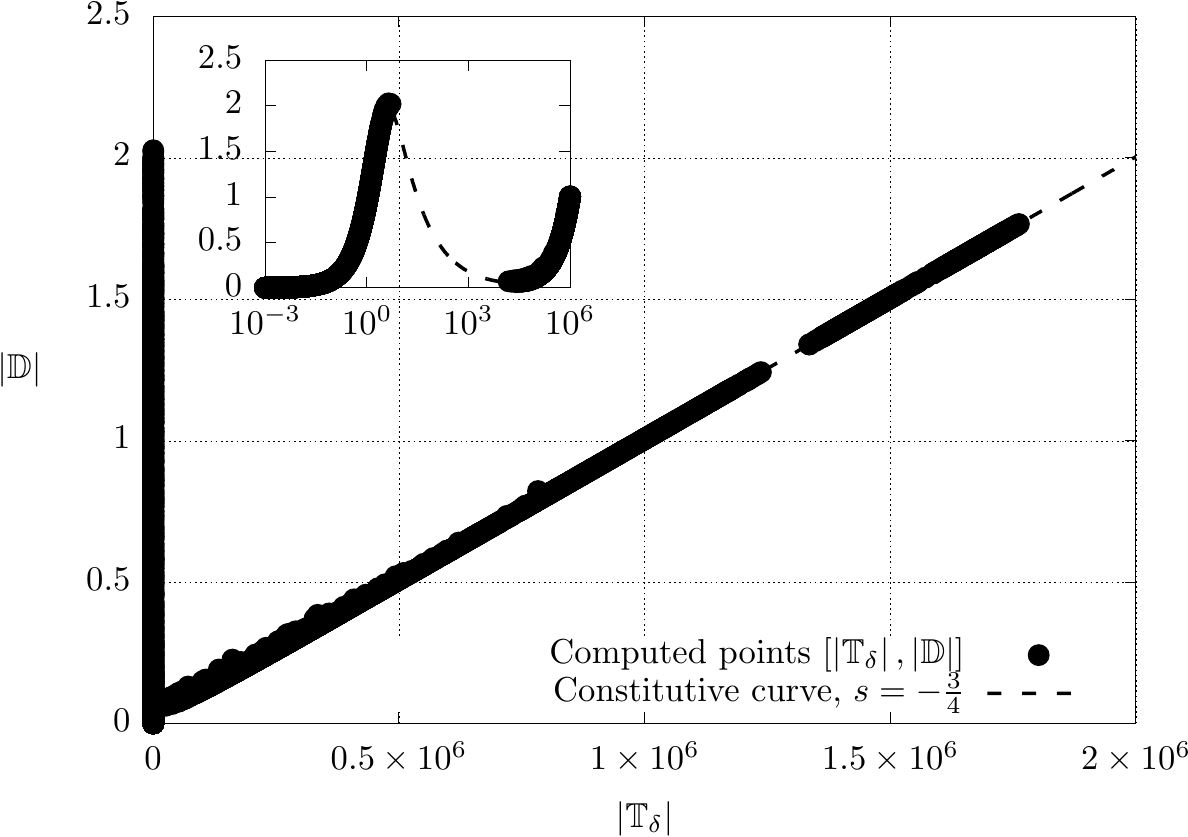}
\caption{Constitutive curve $\absnorm{\gradsym} = \left[\alpha \left(1 + \beta \absnorm{\dcstresssymb}^2 \right)^s + \gamma \right] \absnorm{\dcstresssymb}$ and the computed flux--affinity pairs $[\traceless{\cstress}, \gradsym]$ at time $t = 2 \times 10^{-8}$. Cylindrical Couette flow driven by the the imposed angular velocity $\omega$, see~\eqref{eq:imposed-omega} and~Figure~\ref{fig:angular-velocity}.}
  \label{fig:SR_20}
\end{figure}  


\subsection{Flow through a channel with a narrowing}
\label{sec:flow-through-channel}

Second, we study the behavior of a fluid described by the non-monotone constitutive relation \eqref{eq:4} in a narrowing-channel geometry. The corresponding flow has a strong extensional character, hence it provides a counterpart to the cylindrical Couette flow setting, where the flow is predominantly the shear flow.

\begin{figure}[h]
  \centering
  \includegraphics[width=0.45\textwidth]{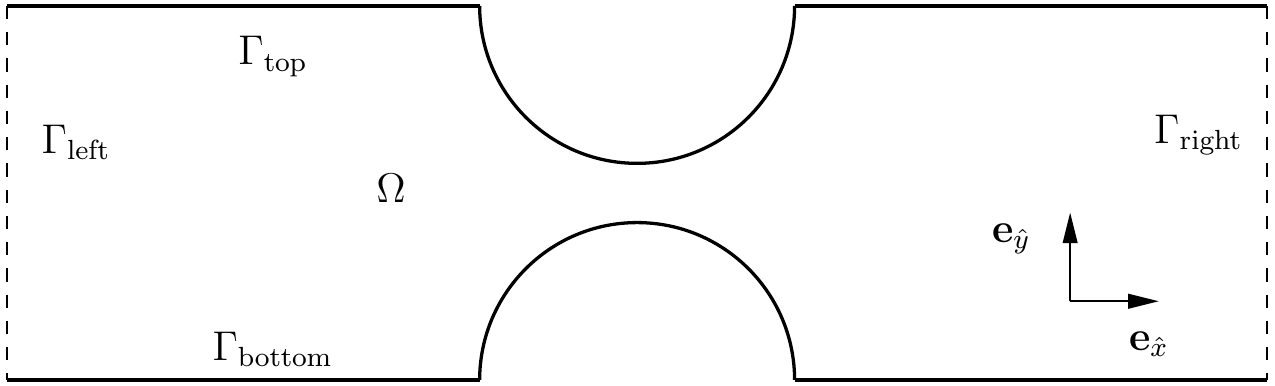}
  \caption{Narrowing channel -- problem geometry.}
  \label{fig:narrowing-channel-geometry}
\end{figure}

The domain being considered is a channel $\Omega = [0,6] \times [0,1]$ with a narrowing at $x = 3$, see Figure~\ref{fig:narrowing-channel-geometry}. The geometry of the narrowing is for $x \in [2.5, 3.5]$ given by the function $0.4 \sin \left( \pi (x - 0.5) \right)$ at the bottom wall and by the function $1 + 0.4 \sin \left( \pi (x + 0.5) \right)$ at the top wall of the channel. The initial condition is a fluid at rest
\begin{equation}
  \label{eq:34}
  \left. {\vecv} \right|_{t=0}= \vec{0}, 
\end{equation}
and we impose the following boundary conditions
\begin{subequations}
  \label{eq:35}
  \begin{align}
    \label{eq:36}
    \left. \vecv \right|_{\Gamma_{\mathrm{top}} \cup \Gamma_{\mathrm{bottom}}} &= \vec{0}, \\
    \label{eq:37}
    \left. \cstress \vec{n} \right|_{\Gamma_{\mathrm{right}}} &= \vec{0}, \\
    \label{eq:38}
    \left. \vecv \right|_{\Gamma_{\mathrm{left}}} &= \begin{bmatrix} f_0 \left( -y^2 + y \right) \\  0 \end{bmatrix} 
  \end{align}
\end{subequations}
where $\Gamma = \Gamma_{\mathrm{top}} \cup \Gamma_{\mathrm{bottom}} \cup \Gamma_{\mathrm{left}} \cup \Gamma_{\mathrm{right}} = \partial \Omega$ represents the boundary of domain $\Omega$ and $f_0>0$ is a constant. The parameters used in the numerical experiments are shown in Table~\ref{tab:narrowing-channel-parameters}. Note that the parameters in the constitutive relation are the same as that used in the cylindrical Couette setting, see Table~\ref{tab:cylindrical-couette-parameters}. The spatial discretisation of the computational domain contained 1800 cells with the minimum cell size $0.021$, and maximum cell size $0.239$. Total number of degrees of freedom (DOF) for the unknown fields was $\mathrm{DOF}_{\text{velocity}} = 3783$, $\mathrm{DOF}_{\text{pressure}} = 992$, $\mathrm{DOF}_{\text{viscosity}} = 1800$.

\begin{table}[ht]
  \centering
  \begin{tabular}{*{7}{c}}
    \toprule
    $\alpha$ & $\beta$ & $\gamma$ & $s$ & $\Delta t$ &$\mathrm{tol}$  \\
    \midrule
    $1.0 $ & $0.1$ & $10^{-6}$ & $-0.75$ & $10^{-10}$ & $10^{-5}$ \\
    \bottomrule
  \end{tabular}
  \caption{Parameters used in the numerical experiments in the narrowing channel flow problem.}
  \label{tab:narrowing-channel-parameters}
\end{table}

\begin{figure}[h]
  \centering
  \subfloat[\label{fig:narrowing-channel-viscosity-a}$f_0=0.001$]{\includegraphics[width=0.6\textwidth]{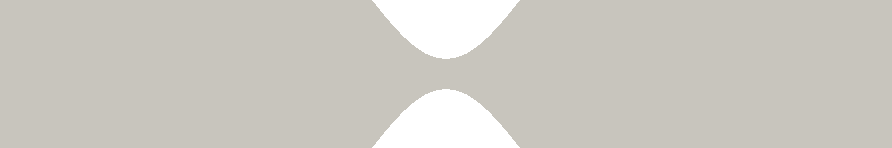}}
  \\
  \subfloat[\label{fig:narrowing-channel-viscosity-b}$f_0=0.01$]{\includegraphics[width=0.6\textwidth]{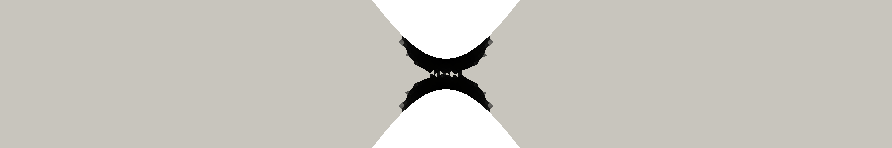}}
  \\
  \subfloat[\label{fig:narrowing-channel-viscosity-c}$f_0=1$]{\includegraphics[width=0.6\textwidth]{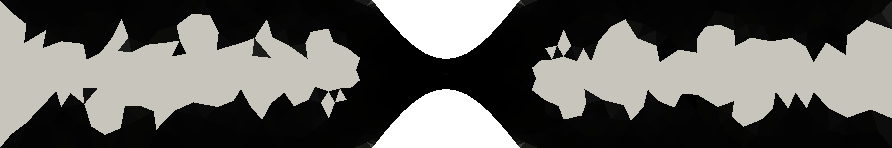}}
  \\
  \subfloat{%
    \includegraphics[width=0.18\textwidth]{taylor-couette-adam/legend/legend-viscosity}}
  \caption{Computed apparent viscosity $\tilde{\mu}$ at $t = 10^{-7}$ in the narrowing channel. Flow is driven by the imposed inlet velocity profile~\eqref{eq:38}.}
  \label{fig:narrowing-channel-viscosity}
\end{figure}

As in the case of cylindrical Couette-flow, we see that if the forcing is small, that is if $f_0 = 0.001$, see Figure~\ref{fig:narrowing-channel-viscosity-a}, then the viscosity is small, which essentially means that all $[\traceless{\cstress}, \gradsym]$ pairs occupy Region 1 on the constitutive curve. As the forcing increases, a high viscosity region starts to appear at the locations with the high values of $\gradsym$, that is in the narrowing of the channel, see Figure~\ref{fig:narrowing-channel-viscosity-b}. Finally, with a strong forcing, the high viscosity regions start to dominate the flow, see Figure~\ref{fig:narrowing-channel-viscosity-c}. The interface between the high viscosity/low viscosity region is again blurry and its detailed features depend on the tolerances in the numerical method and on the quality of the mesh. However, the numerical experiments have again shown that the overall ``averaged'' location of the interface is quite robust with respect to the choice of tolerances in the numerical method as well as on the quality of the mesh. The reader interested in the snapshots of the velocity field and the stress field is referred to~\cite{malek.j.tierra.g:numerical}, the outcomes of the current numerical experiments are qualitatively the same.

\subsection{Comments on numerical experiments}
\label{sec:conclusion-2}

We have designed simple numerical experiments that allowed us to investigate quantitative and qualitative behaviour of a system whose response is described by an implicit constitutive relation~\eqref{eq:4}. As in the case of the reduced model, see Section~\ref{sec:uq}, the solution of the full initial/boundary value problem always contains $[\traceless{\cstress}, \gradsym]$ pairs that occupy either Region~1 or Region~3 on the constitutive curve. The unstable Region 2 is---in the given settings---never occupied by the computed $[\traceless{\cstress}, \gradsym]$ pairs. This is in agreement with the findings by~\cite{janecka.a.pavelka.m:non-convex}.


\section{Conclusion}
\label{sec:conclusion}
A numerical scheme for simulation of transient flows of incompressible non-Newtonian fluids characterised by the non-monotone constitutive relation~\eqref{eq:4} has been proposed. The numerical scheme has been shown to satisfy some rudimentary properties, namely the discretised system of governing equations has been shown to posses a solution, see Section~\ref{sec:scheme}. Using the scheme, we have performed several numerical experiments. The experiments indicate that in the scenarios where the flow is forced by an imposed velocity field, and hence by the imposed shear-rate, then only a portion of the S-shaped curve in the Cauchy stress--symmetric part of the velocity gradient plot is actually active in the complex flows. In particular, the computed flux-affinity pairs $[\traceless{\cstress}, \gradsym]$ have been found to never occupy the decreasing part of the constitutive S-shaped curve, see Figure~\ref{fig:consittutive-relation-scheme-b}.

It has been observed that the flow domain usually splits into multiple regions, while different branches of the constitutive curve are active in the particular regions, see Section~\ref{sec:uq} and Section~\ref{sec:simulations} for details. It is known that such a behaviour might be tantamount to morphological changes in the microscopic constituents of the fluid, see for example~\cite{boltenhagen.p.hu.y.ea:observation} and \cite{hu.yt.boltenhagen.p.ea:shear*1}. Interestingly, such morphological changes can be visualised by various experimental techniques, see for example references in \cite{divoux.t.fardin.ma.ea:shear} and \cite{fardin.m.radulescu.o.ea:stress}, hence the predicted flow induced morphological heterogenity of the fluid is potentially verifiable in experiments.

The \emph{exact} position of the ``mushy'' interface between the high viscosity and low viscosity regions, and consequently between the different branches of the constitutive curve seems to be quite sensitive to non-physical aspects of the problem (numerical parameters). On the other hand, the experimental results also do not lead to a well specified interface as well, see for example~\cite{boltenhagen.p.hu.y.ea:observation} and \cite{hu.yt.boltenhagen.p.ea:shear*1}, the interface is always a bit blurry. An approach that would allow one to better control the position of the interface could be based on the inclusion of the stress diffusion term, which is a popular approach in the mathematical modelling of a closely related shear banding phenomenon, see for example~\cite{divoux.t.fardin.ma.ea:shear} and~\cite{malek.j.prusa.v.ea:thermodynamics}. Such an investigation is however beyond the scope of the current contribution.

Most of the arguments used in the development of the numerical scheme can be also applied to general constitutive relations of the type $\gradsym = g(\absnorm{\dcstresssymb}) \dcstresssymb$, where $g$ is a suitable scalar function, or for that matters, to \emph{any similar constitutive relation between thermodynamic fluxes and affinities}, such as heat flux/temperature gradient, diffusive flux/concentration gradient and so forth. Conceptually, constitutive relation~\eqref{eq:4} belongs to the class of \emph{implicit constitutive relations}, see \cite{rajagopal.kr:on*3,rajagopal.kr:on*4}, \cite{prusa.rajagopal.kr:on}, \cite{perlacova.t.prusa.v:tensorial}, \cite{rajagopal.kr.saccomandi.g:novel} and~\cite{fusi.l.farina.a.ea:lubrication} to name a few, which seems to be an interesting approach to the modelling of fluid response. (See also \cite{bustamante.r:some,bustamante.r.rajagopal.kr:solutions,bustamante.r.rajagopal.kr:on,bustamante.r.rajagopal.kr:implicit,bustamante.r.rajagopal.kr:implicit*1} and~\cite{gokulnath.c.saravanan.u.ea:representations} for a similar developments in the case of solids.)  The presented study opens the possibility to investigate the flows of fluids characterised by implicit constitutive relations in complicated geometries.


\bibliographystyle{chicago}
\bibliography{vit-prusa,bibliography}   

\addtocontents{toc}{\protect\end{multicols}} 
\end{document}